\documentclass[10pt,draftclsnofoot]{IEEEtran}
\onecolumn
\usepackage[centertags]{amsmath}
\usepackage{amssymb,amsthm,cite}
\usepackage[dvips]{graphics}
\usepackage{graphicx}
\usepackage{psfrag}
\usepackage{bbm}
\usepackage{dblfloatfix}
\usepackage{mathtools}
\usepackage{algpseudocode}
\usepackage{algorithm}
\usepackage{hyperref}
\usepackage{color}
\DeclareMathOperator*{\argmax}{\arg\max}

\DeclareMathOperator*{\nn}{\nonumber}
\DeclareMathOperator{\E}{\mathbb{E}}

\allowdisplaybreaks
\newcommand{\RNum}[1]{\uppercase\expandafter{\romannumeral #1\relax}}
\newtheorem{lemma}{Lemma}[section]
\newtheorem{theorem}{Theorem}[section]
\newtheorem{corollary}{Corollary}
\theoremstyle{definition}

\newtheorem{definition}{Definition}
\newtheorem{remark}{Remark}
\makeatletter
\def\blfootnote{\gdef\@thefnmark{}\@footnotetext}
\makeatother

\def\i{\iota}
\def\cJ{{\mathcal J}}
\def\cL{{\mathcal L}}
\def\cM{{\mathcal M}}
\def\cQ{{\mathcal Q}}
\def\cX{{\mathcal X}}
\def\cY{{\mathcal Y}}

\def\sfJ{{\mathsf J}}
\def\tcJ{\widetilde{\cJ}}

\def\ti{\tilde{\iota}}
\def\tmu{\tilde{\mu}}
\def\ts{\tilde{s}}
\def\tildet{\tilde{t}}

\begin{document}
\title{Feedback capacity and coding for the BIBO channel with a no-repeated-ones input constraint}
\author{Oron Sabag, Haim H. Permuter and Navin Kashyap}
\maketitle
\begin{abstract}
In this paper, a general binary-input binary-output (BIBO) channel is investigated in the presence of feedback and input constraints. The feedback capacity and the optimal input distribution of this setting are calculated for the case of an $(1,\infty)$-RLL input constraint, that is, the input sequence contains no consecutive ones. These results are obtained via explicit solution of an equivalent dynamic programming optimization problem. A simple coding scheme is designed based on the principle of posterior matching, which was introduced by Shayevitz and Feder for memoryless channels. The posterior matching scheme for our input-constrained setting is shown to achieve capacity using two new ideas: \textit{history bits}, which captures the memory embedded in our setting, and \textit{message-interval splitting}, which eases the analysis of the scheme. Additionally, in the special case of an S-channel, we give a very simple zero-error coding scheme that is shown to achieve capacity. For the input-constrained BSC, we show using our capacity formula that feedback increases capacity when the cross-over probability is small.
\end{abstract}
\section{Introduction}\label{sec:intro}
\blfootnote{The work of O. Sabag and H. H. Permuter was supported in part by European Research Council under the European Union’s Seventh Framework Programme (FP7/2007-2013)/ERC grant agreement $\mathrm{n}^\mathrm{o}337752$. All authors have also been partially supported by a Joint UGC-ISF research grant. Part of this work was presented at the 2016 International Conference on Signal Processing and Communications (SPCOM 2016), Bangalore, India. O. Sabag and H. H. Permuter are with the Department of Electrical and Computer Engineering, Ben-Gurion University of the Negev, Beer-Sheva, Israel (oronsa@post.bgu.ac.il, haimp@bgu.ac.il). N. Kashyap is with the Department of Electrical Communication Engineering, Indian Institute of Science, Bangalore, India (nkashyap@iisc.ac.in).}
Consider the binary symmetric channel (BSC), described in Fig. \ref{fig:BC} with $\alpha=\beta$, in the presence of output feedback. This setting is well understood in terms of capacity, $C=1-H_2(\alpha)$, but also in terms of efficient and capacity-achieving coding schemes such as the Horstein scheme \cite{horstein_original} and the posterior matching scheme (PMS) \cite{shayevitz_posterior_mathcing}. However, imposing constraints on the input sequence, even in the simplest cases, makes the capacity calculation challenging, since this setting is equivalent to a finite-state channel. A special case of the setting studied here is the BSC with feedback and a no-consecutive-ones input constraint (Fig. \ref{fig:setting}), that is, the channel input sequence cannot contain adjacent ones. We will show for instance, that its feedback capacity still has a simple expression:
\begin{align}\label{eq:BSC_intro}
    C &= \max_{p} \frac{H_2(p) + pH_2\left(\frac{\alpha(1-\alpha)}{p}\right)}{1+p} - H_2(\alpha),
\end{align}
and that there exists an efficient coding scheme that achieves this feedback capacity. It is also interesting to understand the role of feedback on capacity when input constraints are present, and it will be proven that in contrast to the unconstrained BSC, \emph{feedback does increase capacity} for the input-constrained BSC.

The capacity of input-constrained memoryless channels has been extensively investigated in the literature, but still there are no computable expressions for the capacity without feedback \cite{vontobel_generalization,han_constrained_BSC_BEC,wolf_RLL,han_RLL_BSC,yonglong_isit_erasure}. On the other hand, in \cite{Yang05}, it was shown that if there is a noiseless feedback link to the encoder (Fig. \ref{fig:setting}), then the feedback capacity can be formulated as a dynamic programming (DP) problem, for which there exist efficient numerical algorithms for capacity computation \cite{Tatikonda00,TatikondaMitter_IT09,PermuterCuffVanRoyWeissman08,Ising_channel,Sabag_BEC,trapdoor_generalized,Ising_artyom_IT,Chen05}. However, as indicated by the authors of \cite{Yang05}, analytic expressions for the feedback capacity and the optimal input distributions are still hard to obtain and remain an open problem. In this paper, both feedback capacity and the optimal input distribution of the binary-input binary-output (BIBO) channel (Fig. \ref{fig:BC}) with a no-consecutive-ones input constraint are derived by solving the corresponding DP problem. The BIBO channel includes as special cases the BSC ($\alpha=\beta$), which was studied in \cite{Yang05}, the $Z$-channel ($\alpha=0$) and the $S$-channel ($\beta=0$).

Shannon proved that feedback does not increase the capacity of a memoryless channel \cite{shannon56}; following the proof of his theorem, he also claimed that ``feedback does not increase the capacity for channels with memory if the internal channel state can be calculated at the encoder". The input-constrained setting studied here can be cast as a state-dependent channel, so it fits Shannon's description of such a channel. Therefore, we investigate the role of feedback for the special case of input-constrained BSC. In the regime $\alpha \to 0$, the feedback capacity from \eqref{eq:BSC_intro} is compared with a corresponding expression obtained for the capacity without feedback \cite{Han_asymptotics_symmertic}. This comparison reveals that feedback increases capacity, at least for small enough values of $\alpha$ for the input-constrained BSC in contrast to Shannon's claim. However, this is not the first counterexample to Shannon's claim; two other such counterexamples can be found in \cite{Shaviv_shannon_conjecture,Yonglong_increases}.

\begin{figure}[t]
    \centering
    \psfrag{Q}[][][1]{$0$}
    \psfrag{W}[][][1]{$1$}
    \psfrag{A}[][][1]{$\alpha$}
    \psfrag{B}[][][1]{$\beta$}
    \psfrag{C}[][][1]{$1-\alpha$}
    \psfrag{D}[][][.8]{$1-\beta$}
    \psfrag{S}[][][.8]{$1-\alpha$}
    \psfrag{K}[][][1]{$X$}
    \psfrag{L}[][][1]{$Y$}
       \includegraphics[scale=0.6]{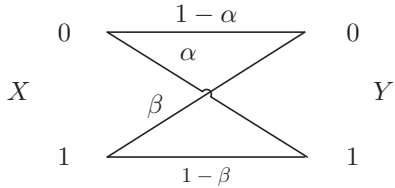}
    \caption{BIBO channel with transition probabilities $(\alpha,\beta)$. Special cases are the Z and S channels, which correspond to $\alpha=0$ and $\beta=0$, respectively, and the BSC when $\alpha=\beta$.}
    \label{fig:BC}
\end{figure}

\begin{figure}[b]
\centering
    \psfrag{A}[b][][.8]{Constrained}
    \psfrag{K}[][][.8]{Encoder}
    \psfrag{B}[][][1]{$p_{Y|X}$}
    \psfrag{C}[][][1]{Decoder}
    \psfrag{D}[][][0.8]{Unit Delay}
    \psfrag{E}[][][1]{$m\in 2^{nR}$}
    \psfrag{F}[][][.9]{$x_i(m,y^{i-1})$}
    \psfrag{G}[][][1]{$y_i$}
    \psfrag{H}[][][1]{$y_i$}
    \psfrag{I}[][][1]{$y_{i-1}$}
    \psfrag{J}[][][1]{$\hat{m}(y^n)$}
    \includegraphics[scale = 0.7]{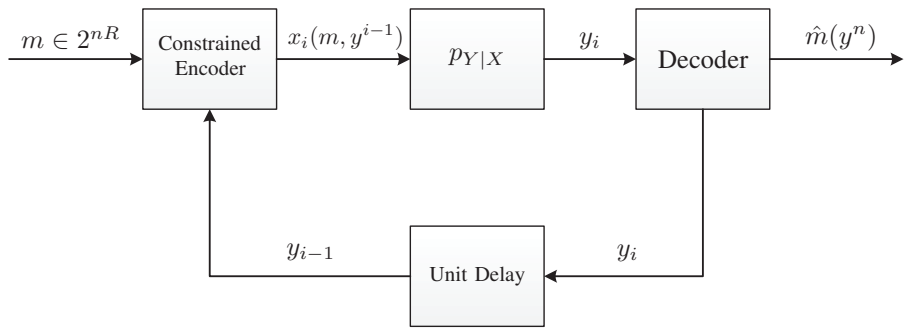}
    \caption{System model for an input-constrained memoryless channel with noiseless feedback.}
    \label{fig:setting}
\end{figure}

In past works on channels with memory, such as \cite{PermuterCuffVanRoyWeissman08,Ising_channel,Sabag_BEC,trapdoor_generalized}, the optimal input distribution provided insights into the construction of simple coding schemes with zero error probability. This methodology also works for the $S$-channel, for which we are able to give a simple zero-error coding scheme. The coding scheme is similar to the "repeat each bit until it gets through" policy that is optimal for a binary erasure channel with feedback. In our case, each bit is repeated with its complement until $Y=0$ is received, so the formed sequence is of alternating bits and satisfies the input constraint. However, a coding scheme for the general BIBO channel is challenging since $p(y|x)>0$, for all $(x,y)$, and therefore, there is no particular pattern of outputs for which a bit can be decoded with certainty. Nonetheless, we are able to use the structure of the optimal input distribution to give a simple coding scheme, based on the principle of posterior matching as is elaborated below.

Two fundamental schemes on sequential coding for memoryless channels with feedback date back to the work of Horstein \cite{horstein_original} for the BSC, and that of Schalkwijk and Kailath \cite{SchalkwijkKailath66_feedback_scheme} for the additive white Gaussian noise (AWGN) channel. In \cite{shayevitz_posterior_mathcing}, Shayevitz and Feder established the strong connection between these coding schemes by introducing a generic coding scheme, termed the \textit{posterior matching scheme} (PMS), for all memoryless channels. This work provided a rigorous proof for the optimality of such sequential schemes, a fact that may be intuitively correct but difficult to prove. Subsequent works proved the coding optimality using different approaches \cite{Coleman_matching,JAVIDI_PMS_IT}, including an original idea by Li and El Gamal in \cite{Li_elgamal_matching} to introduce a randomizer that is available both to the encoder and the decoder. This assumption markedly simplifies the coding analysis, and it was adopted thereafter by \cite{shayevitz_simple} to simplify their original analysis in \cite{shayevitz_posterior_mathcing}. In our coding scheme, it is also assumed that there is a common randomizer available to all parties as a key step to the derivations of an optimal PMS for the BIBO channel.

The encoder principle in the PMS is to determine the channel inputs such that the optimal input distribution is simulated. For a memoryless channel, the optimal input distribution is i.i.d. so the encoder simulates the same experiment at all times. In the input-constrained setting, the input distribution is given by $p(x_i|x_{i-1},y^{i-1})$ (inputs are constrained with probability $1$), so the conditioning injects new information on which the encoder should depend. The first element in the conditioning, $y^{i-1}$, can be viewed as a time-sharing (not i.i.d.) since both the encoder and the decoder know this tuple. Indeed, it is shown that they do not need to track the entire tuple $y^{i-1}$, but a recursive quantization of it on a directed graph. The second element, $x_{i-1}$, is a new element in the PMS since it is only available to the encoder, and it is handled by introducing a new idea called the \textit{history bit} for each message. The analysis of the scheme is simplified using \textit{message-interval splitting}, which results in a homogenous Markov chain instead of a time-dependent random process. These two ideas constitute the core of the PMS for the input-constrained setting, and it is shown that the coding scheme achieves the capacity of the general input-constrained BIBO channel.

The remainder of the paper is organized as follows. Section~\ref{sec:definitions} presents our notation and a description of the problem we consider. Sections~\ref{sec:results} and \ref{sec:optinput} contain statements of the main technical results of the paper. In Section~\ref{sec:scheme}, we provide the PMS for our input-constrained setting, while the optimality of this scheme is proved in Section \ref{sec:proof_scheme}. The DP formulation of feedback capacity together with its solution is presented in Section~\ref{sec:DP_formulation}. Section~\ref{sec:conclusions} contains some concluding remarks. Some of the more technically involved proofs are given in appendices to preserve the flow of the presentation.

\section{Notation and problem definition}\label{sec:definitions}
Random variables will be denoted by upper-case letters, such as $X$, while realizations or specific values will be denoted by lower-case letters, e.g., $x$. Calligraphic letters, e.g., $\mathcal{X}$, will denote the alphabets of the random variables. Let $X^{n}$ denote the $n$-tuple $(X_{1},\dots,X_{n})$ and let $x^n$ denote the realization vectors of $n$ elements, i.e., $x^n = (x_1, x_2, \dots, x_n)$. For any scalar $\alpha\in[0,1]$, $\bar{\alpha}$ stands for $\bar{\alpha}=1-\alpha$. Let $H_2(\alpha)$ denote the binary entropy for the scalar $\alpha\in[0,1]$, i.e., $H_2(\alpha)=-\alpha\log_2\alpha-\bar{\alpha}\log_2\bar{\alpha}$.

The probability mass function (pmf) of a random variable $X$ is denoted by $p_X(x)$, and conditional and joint pmfs are denoted by $p_{Y|X}(y|x)$ and $p_{X,Y}(x,y)$, respectively; when the random variables are clear from the context we use the shorthand notation $p(x)$, $p(y|x)$ and $p(x,y)$. The conditional distribution $p_{Y|X}$ is specified by a stochastic matrix $P_{Y|X}$, the rows of which are indexed by $\cX$, the columns by $\cY$, and the $(x,y)$th entry is the conditional probability $p_{Y|X}(y|x)$ for $x \in \cX$ and $y \in \cY$.

The communication setting (Fig. \ref{fig:setting}) consists of a message $M$ that is drawn uniformly from the set $\{1,\dots,2^{nR}\}$ and made available to the encoder. At time $i$, the encoder produces a binary output, $x_i \in \{0,1\}$, as a function of $m$, and the output samples $y^{i-1}$. The sequence of encoder outputs, $x_1x_2x_3\ldots$, must satisfy the $(1,\infty)$-RLL input constraint, i.e., no consecutive ones are allowed. The transmission is over the BIBO channel (Fig. \ref{fig:BC}) that is characterized by two transition probabilities, $p_{Y|X}(1|0)=\alpha$ and $p_{Y|X}(0|1)=\beta$, where $\alpha$ and $\beta$ are scalars from $[0,1]$. The channel is memoryless, i.e., $p(y_i|x^{i},y^{i-1})=p_{Y|X}(y_i|x_i)$ for all $i$.

\begin{definition} \label{def:code}
A $(n,2^{nR},(1,\infty))$ \emph{code} for an input-constrained channel with feedback is defined by a set of encoding functions:
    \begin{equation*}
        f_i: \{1,\dots,2^{nR}\}\times \mathcal{Y}^{i-1} \rightarrow \mathcal{X}, \ i=1,\dots,n,
    \end{equation*}
    satisfying $f_i(m,y^{i-1})=0 \  \text{if} \ f_{i-1}(m,y^{i-2})=1$ (the mapping $f_1(\cdot)$ is not constrained), for all $(m,y^{i-1})$, and by a decoding function $\Psi: \mathcal{Y}^{n} \rightarrow \{1,\dots,2^{nR}\}$.
\end{definition}

The \textit{average probability of error} for a code is defined as $P_{e}^{(n)}=\Pr[M\neq\Psi(Y^{n})]$. A rate $R$ is said to be $(1,\infty)$\textit{-achievable} if there exists a sequence of $(n,2^{nR},(1,\infty))$ codes such that $\lim_{n\rightarrow\infty} P_{e}^{(n)}=0$. The \textit{capacity}, $C^{\mathrm{fb}}(\alpha,\beta)$ is defined as the supremum over all $(1,\infty)$-achievable rates.

The transition probabilities can be restricted to $\alpha+\beta \leq 1$, a fact that is justified by:
\begin{lemma}
The capacity of a BIBO channel satisfies
  $C(\alpha,\beta)=C(1-\alpha,1-\beta)$,
for all $\alpha,\beta$.
\end{lemma}
\begin{proof}
For a channel with parameters $(\alpha,\beta)$, apply an invertible mapping $\tilde{Y}=Y\oplus1$ on channel outputs so that the capacity remains the same but the parameters are changed to $(1-\alpha,1-\beta)$.
\end{proof}
The proof of the lemma is valid even when the inputs are constrained and there is feedback to the encoder.
\section{Main Results}\label{sec:results}
In this section, we present our main results concerning the feedback capacity of the BIBO channel,  and thereafter, we show that feedback increases capacity for the BSC. The optimal PMS for the BIBO channel is not included in this section and appears in Section \ref{sec:scheme}.
\subsection{Feedback capacity}\label{subsec:capacity}
The general expression for the feedback capacity is given by the following theorem
\begin{theorem}[BIBO capacity]\label{theorem:BIBO}
The feedback capacity of the input-constrained BIBO channel is
\begin{align}\label{eq:BIBO_capacity}
    C^{\mathrm{fb}}(\alpha,\beta) &= \max_{z_L\leq z \leq z_U} \frac{H_2(\alpha\bar{z} + \bar{\beta} z) + (\alpha\bar{z} + \bar{\beta} z)H_2\left(\frac{\alpha\bar{\beta}}{\alpha\bar{z} + \bar{\beta} z}\right) - (\bar{z} + \bar{\beta}z)H_2(\alpha) - (z + \alpha\bar{z})H_2(\beta)}{1 + \alpha\bar{z} + \bar{\beta} z},
\end{align}
where $\alpha + \beta \leq 1$, $z_L=\frac{\sqrt{\alpha}}{\sqrt{\alpha}+\sqrt{\bar{\beta}}}$ and $z_U=\frac{\sqrt{\bar{\alpha}}}{\sqrt{\bar{\alpha}}+\sqrt{\beta}}$.\\
The feedback capacity can also be expressed by:
\begin{align}\label{eq:BIBO_alternative}
    C^{\mathrm{fb}}(\alpha,\beta) &= \log\left(\frac{1-p_{\alpha,\beta}}{p_{\alpha,\beta}-\alpha\bar\beta}\right) + \beta\frac{H_2(\alpha)}{1-\alpha-\beta}- \bar\alpha\frac{H_2(\beta)}{1-\alpha-\beta} ,
\end{align}
where $p_{\alpha,\beta}$ is the unique solution of
\begin{align}\label{eq:BIBO_arg_equation}
(1-\alpha\bar{\beta}) [H_2(\alpha)-H_2(\beta)] + (\bar{\beta}-\alpha)[2\log (1-p) - \log(p-\alpha\bar{\beta})(1+\alpha\bar{\beta}) + \alpha\bar{\beta}\log \alpha\bar{\beta}]&=0.
\end{align}
\end{theorem}
The proof of \eqref{eq:BIBO_capacity} in Theorem \ref{theorem:BIBO} appears in Section \ref{sec:DP_formulation} and relies on the formulation of feedback capacity as a DP problem. From the solution of the DP, we only obtain that the maximization in \eqref{eq:BIBO_capacity} is over $z \in [0,1]$, but this can be strengthened using the following result:
\begin{lemma}\label{lemma:maximization_domain}
Define $R_{\alpha,\beta}(z) = \frac{H_2(\alpha\bar{z} + \bar{\beta} z) + (\alpha\bar{z} + \bar{\beta} z)H_2\left(\frac{\alpha\bar{\beta}}{\alpha\bar{z} + \bar{\beta} z}\right) - (\bar{z} + \bar{\beta}z)H_2(\alpha) - (z + \alpha\bar{z})H_2(\beta)}{1 + \alpha\bar{z} + \bar{\beta} z}$ with $0\le z\le 1$. The argument that achieves the maximum of $R_{\alpha,\beta}(z)$ lies within $[z_L,z_U]$, for all $\alpha+\beta\le1$. Additionally, for the BSC ($\alpha=\beta$), the maximum is attained when the argument is within $[z_L,0.5]$.
\end{lemma}

The proof of Lemma \ref{lemma:maximization_domain} appears in Appendix \ref{app:maximization_domain}. The alternative capacity expression \eqref{eq:BIBO_alternative} is obtained by taking the derivative of \eqref{eq:BIBO_capacity} and substituting the resulting relation into the capacity expression \eqref{eq:BIBO_capacity}. Note that the LHS of \eqref{eq:BIBO_arg_equation} is a decreasing function of $p$, and hence, efficient methods can be applied to calculate \eqref{eq:BIBO_alternative}.

\begin{remark}
The feedback capacity can also be calculated using upper and lower bounds from \cite{Sabag_UB_IT}, which turn out to meet for this channel, instead of the DP approach that is taken in this paper.
\end{remark}
Theorem \ref{theorem:BIBO} provides the capacity of three special cases: the BSC, the S-channel and the Z-channel. Their feedback capacities are calculated by substituting their corresponding parameters in Theorem \ref{theorem:BIBO}.
\begin{corollary}[BSC capacity]\label{coro:BSC}
The feedback capacity of the input-constrained BSC $(\alpha=\beta)$ is
\begin{align}\label{eq:BSC_capacity}
    C^{\mathrm{BSC}}(\alpha) &= \max_{\sqrt{\alpha\bar\alpha}\leq p \le 0.5} \frac{H_2(p) + pH_2\left(\frac{\alpha\bar{\alpha}}{p}\right)}{1+p} - H_2(\alpha),
\end{align}
where $\alpha\leq0.5$. An alternative capacity expression is
\begin{align}\label{eq:simple_capacity}
    C^{\mathrm{BSC}}(\alpha) &= \log \left(\frac{1-p_\alpha}{p_\alpha-\alpha\bar{\alpha}}\right) - H_2(\alpha),
\end{align}
where $p_\alpha$ is the unique solution of $(\alpha\bar{\alpha})^{\alpha\bar{\alpha}}(1-p)^2 = (p-\alpha\bar{\alpha})^{1+\alpha\bar{\alpha}}$.
\end{corollary}

By operational considerations, the feedback capacity in Theorem \ref{theorem:BIBO} serves as an upper bound for the non-feedback setting, which is still an open problem. For the BSC, it will be shown further in Theorem \ref{theorem:increases} that feedback increases capacity, at least for small values of $\alpha$, so this upper bound is not tight.

\begin{corollary}[S-channel capacity]
The feedback capacity of the input-constrained S-channel $(\beta=0)$ is
\begin{align}\label{eq:S_capa}
    C^{\mathrm{S}}(\alpha) &= \max_{ \sqrt{\alpha} \leq p \leq 1} \frac{H_2(p) + pH_2\left(\frac{\alpha}{p}\right)- H_2(\alpha)}{1+ p}\nn\\
                  &= \max_{ \sqrt{\alpha} \leq p \leq 1} \bar{\alpha}\frac{H_2\left(\frac{1-p}{1-\alpha}\right)}{1+ p}
\end{align}
The capacity can also be expressed by:
   \begin{align}\label{eq:S_capa_alter}
    C^{\mathrm{S}}(\alpha) &= \log\left(\frac{1-p_\alpha}{p_\alpha-\alpha}\right),
\end{align}
where $p_\alpha$ is the unique solution of $(1-p)^2 = (p-\alpha)^{1+\alpha}\bar{\alpha}^{\bar{\alpha}}$.
\end{corollary}

The second capacity expression in \eqref{eq:S_capa} reveals a simple zero-error coding scheme for the S-channel. To describe this, we first fix a $z \in (0,1)$, and consider a set, $\cM$, consisting of $|\cM| = 2^{NH_2(z)}$ messages, where $N$ is a large integer.\footnote{We will be slightly loose in our description of this coding scheme so as to keep the focus on the simplicity of the scheme. We will ignore all $o_N(1)$ correction terms needed to make our arguments mathematically precise. Thus, for example, we implicitly assume that $2^{NH_2(z)}$ is an integer. We will also assume that $Nz$ is an integer, and that there are $2^{NH(z)}$ binary sequences of length $N$ which contain exactly $Nz$ $1$s.} The coding scheme operates in two stages:
 \begin{enumerate}
 \item \underline{Message shaping}: The set of messages, $\cM$, is mapped in a one-to-one fashion into the set of length-$N$ binary sequences containing $Nz$ $1$s. Thus, each message is identified with a binary sequence of length $N$ with fraction of $1$s equal to $z$. This ``message shaping'' can be implemented, for instance, using the enumerative source coding technique \cite[Example 2]{Cover73}.

\item \underline{Message transmission}: Each of the $N$ bits, $b_1,b_2,\ldots,b_N$, in the shaped message sequence is transmitted by the encoder using the following procedure:
\begin{quote}
To send the message bit $b$, the encoder transmits the sequence $b \, \overline{b} \,  b \, \overline{b} \ldots$, where $\overline{b}$ denotes the complement (NOT) of $b$, until a $0$ is received at the channel output, at which point the transmission is stopped.
\end{quote}
Note that if $y_1 y_2 \ldots y_{\ell-1} 0$ is the sequence received at the S-channel output in response to the transmission of the message bit $b$, then the decoder can determine whether $b = 0$ or $b = 1$ from the parity of $\ell$: if $\ell$ is odd, then $b = 0$; if $\ell$ is even, then $b = 1$.
\end{enumerate}

By the law of large numbers, the number of $S$-channel uses needed for the transmission of an $N$-bit shaped message sequence is close to $N \times \E[L_z]$, where $\E[L_z]$ denotes the expected number of transmissions needed for sending a single $\text{Bernoulli}(z)$ bit $b$ using the procedure described above.
It is easy to check that $\E[L_z]$ equals
\begin{align}\label{eq:S_channel_uses}
  \bar{z}\bar{\alpha}\sum_{k=1}^\infty (2k-1)\alpha^{k-1} + z\bar{\alpha}\sum_{k=1}^\infty 2k \alpha^{k-1} &= -\bar{z} + \bar{\alpha}\sum_{k=1}^\infty 2k\alpha^{k-1} \nn\\
  &=\frac{1+p}{1-\alpha},
\end{align}
where $p=z+\alpha\bar{z}$. Thus, the rate achieved by this scheme is (arbitrarily close to) $\frac{\log|\cM|}{N\E[L_z]} = \frac{H_2(z)}{\E[L_z]} = \frac{1-\alpha}{1+p} \, H_2\left(\frac{1-p}{1-\alpha}\right)$.
Maximizing over $z \in (0,1)$, we conclude that the scheme achieves the S-channel capacity given by \eqref{eq:S_capa}.

\begin{corollary}[Z-channel capacity]
The feedback capacity of the input-constrained Z-channel $(\alpha=0)$ is
\begin{align}\label{main:Z_capacity}
    C^{\mathrm{Z}}(\beta) &= \max_{0\leq p \le \bar{\beta}} \frac{H_2(p) - \frac{p}{1-\beta}H_2(\beta)}{1+p} \nn\\
&= - \log(1-p_\beta),
\end{align}
where $p_\beta$ is the unique solution of the quadratic equation $(1-p)^2 = p \cdot 2^{\frac{H_2(\beta)}{1-\beta}}$.
\end{corollary}
The feedback capacities of the input-constrained S and Z channels are different  because of the asymmetry imposed by the input constraint (Fig. \ref{fig:ZS}). Note that for most values of the channel parameters, the capacity of the $S$-channel exceeds that of the $Z$-channel; intuitively, the decoder can gain more information when observing two consecutive ones in the channel output because it knows that there is one error in this transmission pair.

\begin{figure}[h]
\centering
        \centerline{\includegraphics[scale=0.4]{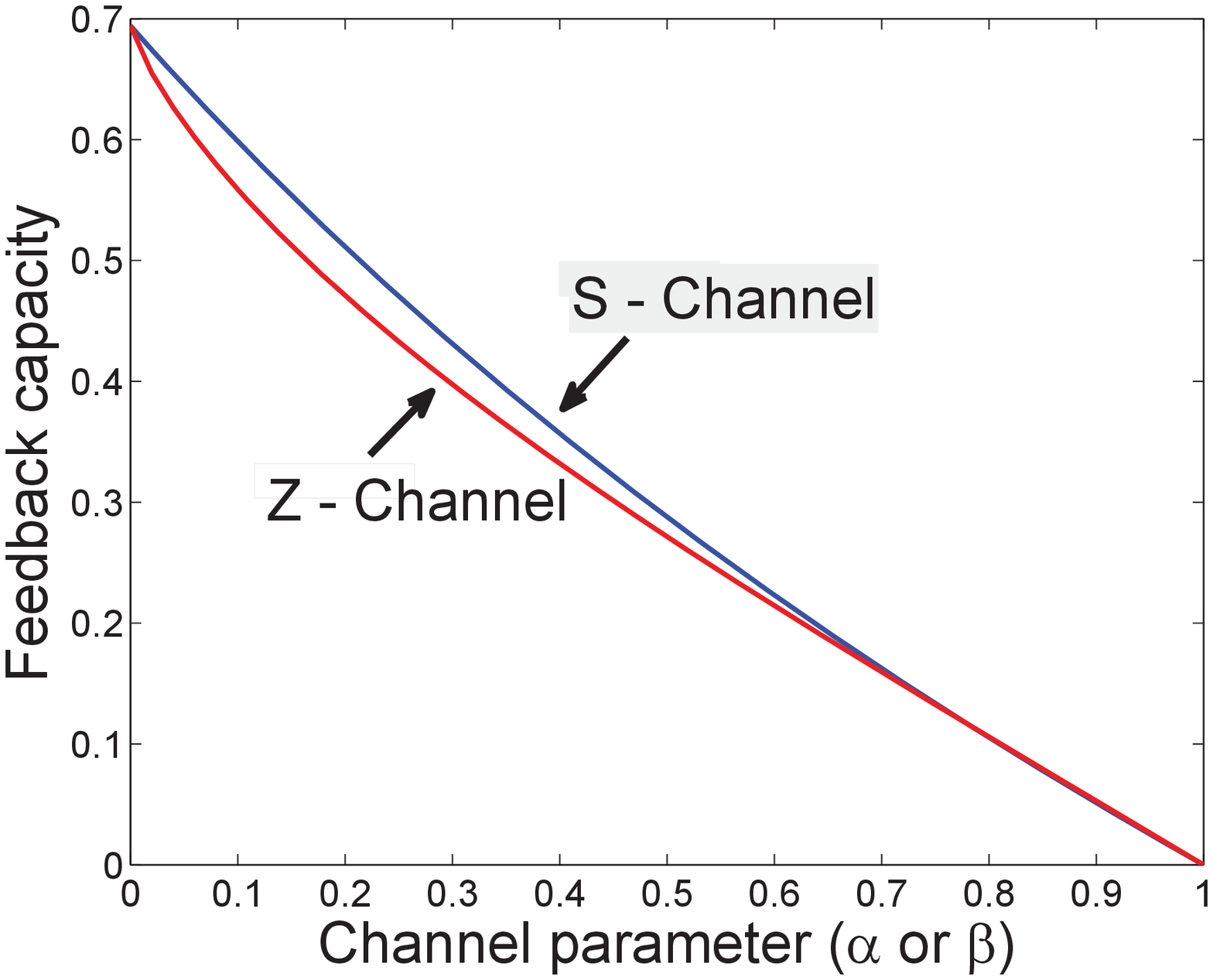}}
\caption{Comparison between the capacities of the constrained Z- and S- channels.}
\label{fig:ZS}
\vspace{-6mm}
\end{figure}
\subsection{Feedback increases capacity}
In this section, we show that feedback increases capacity for the input-constrained BSC.
\begin{theorem}\label{theorem:increases}
Feedback increases capacity for the $(1,\infty)$-RLL input-constrained BSC, for all values of $\alpha$ in some neighborhood around $0$.
\end{theorem}
As discussed in Section I, this gives a counterexample to a claim of Shannon's from \cite{shannon56}. A subsequent work \cite{andrew_upperbounds} related to the conference version of our paper \cite{sabag_allerton_15} used a novel technique to calculate upper bounds on the non-feedback capacity of the input-constrained BSC. The upper bound in \cite{andrew_upperbounds} is a tighter upper bound than our feedback capacity, which shows that feedback increases capacity not only for small values of $\alpha$, but actually for all $\alpha$.

In order to show Theorem \ref{theorem:increases}, we provide the asymptotic expressions of the input-constrained BSC with and without feedback.
\begin{theorem}\label{theorem:asymptotic}
The feedback capacity of the input-constrained BSC is:
\begin{equation}\label{eq:asymptotic}
 C^{\mathrm{BSC}}(\alpha) = \log \lambda + \frac{2-\lambda}{3-\lambda}\alpha\log\alpha  + \left( \frac{\log(2-\lambda)-(2-\lambda)}{3-\lambda}\right)\alpha + O(\alpha^2\log^2\alpha),
\end{equation}
where $\lambda$ is the golden ratio $\left(\lambda=\frac{1+\sqrt{5}}{2}\right)$.
\end{theorem}
The derivation of Theorem \ref{theorem:asymptotic} is more involved than standard Taylor series expansion about $\alpha = 0$, since the second-order term of \eqref{eq:asymptotic} is $O(\alpha\log\alpha)$. The proof of Theorem \ref{theorem:asymptotic} appears in Appendix \ref{app:shannon}.

The asymptotic behaviour of the capacity of the input-constrained BSC without feedback is captured by the following result.
\begin{theorem}\cite[Example 4.1]{Han_asymptotics_symmertic}\label{theorem_Han}
The non-feedback capacity of the $(1,\infty)$-RLL input-constrained BSC is:
\begin{equation}\label{eq:theorem_HAN}
C^{\mathrm{NF}}(\alpha) = \log \lambda + \frac{2\lambda+2}{4\lambda+3}\alpha\log\alpha + O(\alpha).
\end{equation}
\end{theorem}
It is now easy to prove Theorem \ref{theorem:increases}.
\begin{proof}[Proof of Theorem \ref{theorem:increases}]
The coefficients of the term $\alpha\log \alpha$ in \eqref{eq:asymptotic} and \eqref{eq:theorem_HAN} satisfy $\frac{2\lambda+2}{4\lambda+3}>\frac{2-\lambda}{3-\lambda}$. Therefore, there exists $\alpha^\ast>0$ such that
$C^{\mathrm{BSC}}(\alpha)- C^{\mathrm{NF}}(\alpha)>0$, for all $\alpha<\alpha^\ast$.
\end{proof}

\section{The Optimal Input Distribution}\label{sec:optinput}
In this section, we present the optimal input distribution for the input-constrained BIBO, based on which the capacity-achieving coding scheme of the next section is derived.
The optimization problem that needs to be solved when calculating the feedback capacity of our setting is given in the following theorem.
\begin{theorem}[\cite{Sabag_BEC}, Theorem $3$]\label{theorem:cap_as_DP}
The capacity of an $(1,\infty)$-RLL input-constrained memoryless channel with feedback can be written as:
\begin{equation}\label{eq_capacity_as_DP}
  C^{\mathrm{fb}} = \sup \liminf_{N\rightarrow \infty} \frac{1}{N} \sum_{t=1}^{N}I(X_t;Y_t|Y^{t-1}),
\end{equation}
where the supremum is taken with respect to $\{p_{X_{t}|X_{t-1},Y^{t-1}}:p_{X_{t}|X_{t-1},Y^{t-1}}(1|1,y^{t-1})=0\}_{t\geq 1}$.
\end{theorem}
The input at time $t$ depends on the previous channel input, $x_{t-1}$, and the output samples $y^{t-1}$. The description of such an input distribution is difficult since the conditioning contains a time-increasing domain, $\mathcal{Y}^{t-1}$. The essence of the DP formulation is to replace the conditioning on $y^{t-1}$ with $p_{X_{t-1}|Y^{t-1}}(0|y^{t-1})$, which is a sufficient statistic of the outputs tuple. Furthermore, the DP solution in Section \ref{sec:DP_formulation} reveals that the \textit{optimal input distribution} can be described with an even simpler notion called a $Q$-graph, which is suitable for scenarios where the DP state, $p_{X_{t-1}|Y^{t-1}}(0|y^{t-1})$, takes a finite number of values.

\begin{figure}[t]
\centering
        \psfrag{A}[][][1.2]{$Q=3$}
        \psfrag{B}[][][1.2]{$Q=2$}
        \psfrag{C}[][][1.2]{$Q=1$}
        \psfrag{D}[][][1.2]{$Q=4$}
        \psfrag{E}[][][1]{$Y=0$}
        \psfrag{G}[][][1]{$p^{\text{opt}}_\alpha$}
        \psfrag{G}[][][1]{}
        \psfrag{I}[][][1]{$1-q^{\text{opt}}_\alpha$}
        \psfrag{I}[][][1]{}
        \psfrag{H}[][][1]{$q^{\text{opt}}_\alpha$}
        \psfrag{H}[][][1]{}
        \psfrag{E}[][][1]{$Y=1$}
        \psfrag{F}[][][1]{$Y=0$}
        \psfrag{I}[][][1]{$Y=0$}
        \psfrag{H}[][][1]{$Y=1$}
      \psfrag{Z}[][][.8]{}
      \psfrag{X}[][][.8]{}
      \psfrag{Q}[][][.8]{}
        \centerline{\includegraphics[scale=.4]{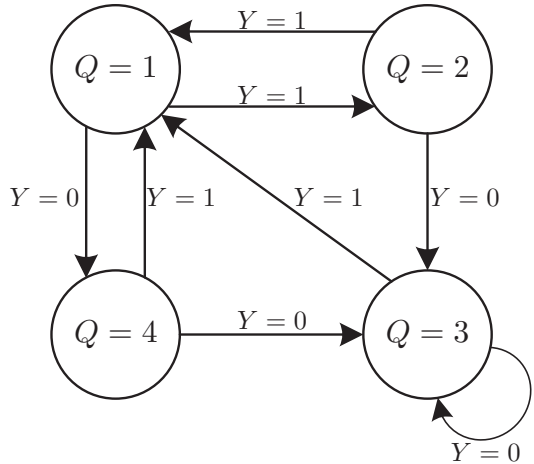}}
\caption{The $Q$-graph that characterizes the optimal input distribution.}
\label{fig:BC_optimal_input}
\end{figure}
\begin{definition}
For an output alphabet, $\mathcal{Y}$, a \textit{$Q$-graph} is a directed, connected and labeled graph. Additionally, each node should have $|\mathcal{Y}|$ outgoing edges, with distinct labels.
\label{def:Qgraph}
\end{definition}

The $Q$-graph depicted in Fig. \ref{fig:BC_optimal_input} will be used to describe the optimal input distribution. Let $\cQ = \{1,2,3,4\}$ denote the set of nodes of this $Q$-graph. We will use a function $g:\cQ \times \{0,1\} \to \cQ$ to record the transitions along the edges of the graph. Specifically, $g(1,0) = 4$, $g(1,1)=2$, $g(2,0) = 3$, $g(2,1)=1$, $g(3,0)=3$, $g(3,1)=1$, $g(4,0) = 3$, and $g(4,1)=1$. Given some initial node $q_0 \in \cQ$ and an output sequence $y^t \in \cY^t$ of arbitrary length, a unique node $q_t \in \cQ$ is determined by a walk on the $Q$-graph starting at $q_0$ and following the edges labeled by $y_1,y_2,\ldots,y_t$, in that order. We will write this as $q_t = \Phi(q_0,y^t)$, where $\Phi:\cQ \times \bigcup_{t \ge 1} \cY^t  \to \cQ$ is the mapping recursively described by $\Phi(q_0, y_1) = g(q_0,y_1)$, and $\Phi(q_0,y^t)=g(\Phi(q_0,y^{t-1}),y_t)$ for $t \ge 2$. The importance of the $Q$-graph for our scheme is that the encoder and decoder need only track the value $\Phi(q_0,y^{t-1})$, instead of the entire output sequence $y^{t-1}$.

For the description of the optimal input distribution, define
\begin{align}\label{eq:BC_capacity}
    z^{\alpha,\beta}_2 &= \argmax_{0\leq z \leq 1} \frac{H_2(\alpha\bar{z} + \bar{\beta} z) + (\alpha\bar{z} + \bar{\beta} z)H_2\left(\frac{\alpha\bar{\beta}}{\alpha\bar{z} + \bar{\beta} z}\right) - (\bar{z} + \bar{\beta}z)H_2(\alpha) - (z + \alpha\bar{z})H_2(\beta)}{1 + \alpha\bar{z} + \bar{\beta} z},
\end{align}
with the following subsequent quantities:
\begin{equation}
\label{eq:def_Zi}
\begin{split}
    z^{\alpha,\beta}_1&\triangleq \frac{\alpha\bar{z}^{\alpha,\beta}_2}{\alpha\bar{z}^{\alpha,\beta}_2 + \bar{\beta}{z}^{\alpha,\beta}_2}\\
    z^{\alpha,\beta}_3&\triangleq \frac{\bar{\alpha}\bar{z}^{\alpha,\beta}_2}{\bar{\alpha}\bar{z}^{\alpha,\beta}_2 + \beta z^{\alpha,\beta}_2}\\
    z^{\alpha,\beta}_4&\triangleq \frac{\bar{\alpha}\bar{\beta}z^{\alpha,\beta}_2}{\bar{\alpha}\bar{\beta}{z}^{\alpha,\beta}_2 + \alpha\beta \bar{z}^{\alpha,\beta}_2}.
\end{split}
\end{equation}
It can be shown that $z^{\alpha,\beta}_1\leq z^{\alpha,\beta}_2\leq z^{\alpha,\beta}_3\leq z^{\alpha,\beta}_4$ for all $\alpha+\beta\le1$. For instance, the relation $z_1\leq z_2$ (superscripts $(\alpha,\beta)$ are omitted) can be simplified to $(\bar{\beta}-\alpha)z^2_2 + 2\alpha z_2 -\alpha\ge0$. Now, the polynomial $(\bar{\beta}-\alpha)x^2 + 2\alpha x -\alpha$ has two roots, one is negative and the other is at $x=z_L$. Since the polynomial is convex, $(\bar{\beta}-\alpha)z^2_2 + 2\alpha z_2 -\alpha\ge0$ is equivalent to $z_2\ge z_L$. Using the same methodology, it can be shown that $z_1 \le z_2 \le z_3 \le z_4$ is equivalent to $z_L\le z_2\le z_U$, which is proved in Lemma \ref{lemma:maximization_domain}.

Define the conditional distributions $p^*_{X|X^-,Q}$ via the conditional probability matrices
\begin{equation}
\begin{split}
  p^\ast_{X|X^-, Q = 1} &=
   \left[
                   \begin{array}{cc}
                     0 & 1 \\
                     1 & 0 \\
                   \end{array}
                 \right] \\
  p^\ast_{X|X^-, Q = 2} &= p^\ast_{X|X^-, Q = 1} \\
  p^\ast_{X|X^-, Q = 3} &=  \left[
                   \begin{array}{cc}
                     1-\frac{z^{\alpha,\beta}_2}{z^{\alpha,\beta}_3} & \frac{z^{\alpha,\beta}_2}{z^{\alpha,\beta}_3} \\
                     1 & 0 \\
                   \end{array}
                 \right]  \\
  p^\ast_{X|X^-, Q = 4} &=  \left[
                   \begin{array}{cc}
                     1-\frac{z^{\alpha,\beta}_2}{z^{\alpha,\beta}_4} & \frac{z^{\alpha,\beta}_2}{z^{\alpha,\beta}_4} \\
                     1 & 0 \\
                   \end{array}
                 \right],
\end{split}
\label{eq:transfer_matrices}
\end{equation}
in which $X^-$ indexes the rows and $X$ indexes the columns. To be precise, the first (resp.\ second) row of each matrix is a conditional pmf of $X$ given $X^- = 0$ (resp.\ $X^-=1$).

The optimal input distribution and alternative capacity expression are given in the following theorem:
\begin{theorem}[Optimal input distribution]\label{theorem:optimal_inputs}
For any $q_0 \in \cQ$, the input distribution $p_{X_i \mid X_{i-1},Y^{i-1}}(x \mid x^-,y^{i-1})= p^*_{X \mid X^-,Q}\bigl(x \mid x^-,\Phi(q_0,y^{i-1})\bigr)$, defined via \eqref{eq:transfer_matrices} and Fig. \ref{fig:BC_optimal_input}, is capacity-achieving. Moreover, the random process $\{(X_i,Q_i)\}_{i\geq1}$ induced by $p^\ast_{X|X^-,Q}$ is an irreducible and aperiodic Markov chain on $\{0,1\} \times \cQ$. The stationary distribution of this Markov chain is given by $\pi_{X^-,Q} = \pi_Q \pi_{X^-|Q}$, where $\pi_Q$ is the pmf on $\cQ$ defined by $[\pi_Q(1),\pi_Q(2),\pi_Q(3),\pi_Q(4)] = \left[ \frac{p}{1 + p}, \frac{pq}{1 + p}, \frac{1 - p}{1 + p}, \frac{p(1-q)}{1+p}\right]$, with $p = \alpha (1-z_2^{\alpha,\beta}) + (1-\beta) z_2^{\alpha,\beta}$ and $q = \frac{\alpha(1-\beta)}{p}$, and $\pi_{X^-|Q}(0|i) = 1-\pi_{X^-|Q}(1|i) = z^{\alpha,\beta}_i$ for $i=1,\dots,4$. The feedback capacity $C^{\mathrm{fb}}(\alpha,\beta)$ can be expressed as $I(X;Y|Q)$, where the joint distribution is $\pi_{Q,X,Y}(q,x,y) = \sum_{x^-}p_{Y|X}(y|x)p_{X|X^-,Q}^\ast(x|x^-,q)\pi_{X^-,Q}(x^-,q)$.
\end{theorem}
The scheme uses the joint probability distribution $p_{Y|X}p^\ast_{X|X^-,Q}\pi_{X^-,Q}$ induced by the optimal input distribution $p^\ast_{X|X^-,Q}$. Here, $X$ and $X^-$ should be viewed as the channel inputs during the current and previous time instances, respectively, and $Q$ is the value of the node on the $Q$-graph prior to the transmission of $X$. In the analysis of the coding scheme, we will use the Markov property of $\{(X_i,Q_i)\}_{i\ge1}$ to show that $I(X;Y|Q)$ is achievable. The proof of Theorem \ref{theorem:optimal_inputs} is presented at the end of Section \ref{sec:DP_formulation}.

\section{The Coding Scheme}\label{sec:scheme}

The coding scheme we describe here consists of two phases: Phase~I is based on a posterior matching scheme (PMS), and Phase~II is a clean-up phase based on block codes.

The main element of any PMS is the posterior distribution of the message given the channel outputs. The posterior distribution is represented by the lengths of sub-intervals that form a partition of the unit interval. Each sub-interval is associated with a particular message, and henceforth, it will be referred to as a ``message interval''. The initial lengths are equal for all message intervals, since the decoder is assumed to have no prior information about the messages. The lengths of the message intervals are updated throughout the transmission based on the outputs that are made available to the decoder (and to the encoder from the feedback). The encoder's job is to refine the decoder's knowledge about the correct message by simulating samples from a desired input distribution. When this is done properly, as time progresses, the length of the true message interval will increase towards $1$, and the decoder can then successfully declare its estimate of the correct message.

The above description of PMS is generic and applies to any setting of channel coding with feedback. What is specific to our PMS in the input-constrained setting is the input distribution that the encoder attempts to simulate during its operation. Most of the adaptations needed for our PMS that are not present in the baseline PMS for memoryless channels in \cite{shayevitz_posterior_mathcing} and \cite{Li_elgamal_matching} are a natural consequence of the input constraints and the structure of the input distribution in Theorem~\ref{theorem:optimal_inputs}. However, these adaptations complicate the analysis of the scheme as the evolution in time of the involved random variables results in a random process that is difficult to analyze.

The analysis becomes easier upon introducing a certain message-interval splitting operation (described in Section~\ref{sec:recursion}) that induces a Markov chain structure on the time-evolution of the random variables in the scheme. However, the splitting operation prevents the length of the correct message interval from increasing to $1$, but we will show that, with high probability, the length of this interval will eventually go above some positive constant. The PMS output at the end of Phase~I will be a list of messages whose interval lengths are above this constant.

In Phase~II, a fixed-length block coding scheme, asymptotically of zero rate, is used to determine which message in the list produced at the end of Phase~I is the correct message.

The remainder of this section is organized as follows. In Section~\ref{sec:phase1}, the key elements of the PMS of Phase~I are described. This is followed by a description of Phase~II of our coding scheme in Section~\ref{sec:phase2}. The overall coding scheme that combines the two phases is  shown to be capacity-achieving in Section~\ref{sec:combine}. At the heart of the PMS of Phase~I is a recursive construction of message intervals, which is described in detail in Section~\ref{sec:recursion}. This technical description has been left to the end so as not to distract the reader from the main ideas of the coding scheme.

\subsection{Phase~I: PMS} \label{sec:phase1}
The PMS is based on the joint probability distribution $\pi_{Y,X,X^-,Q}$ on $\cY \times \cX \times \cX \times \cQ$ defined by $\pi_{Y,X,X^-,Q} := \pi_{X^-,Q} p^*_{X|X^-,Q}p_{Y|X}$, where $\pi_{X^-,Q}$ and $p^*_{X|X^-,Q}$ constitute the optimal input distribution described in Section~\ref{sec:optinput}, and $p_{Y|X}$ is the channel law. In what is to follow, we routinely use notation such as $\pi_Q$, $\pi_{X^-|Q}$, $\pi_{X,X^-|Q}$ etc. to denote certain marginal and conditional probability distributions specified by $\pi_{X,X^-,Q}$. Thus, for example, $\pi_Q$ denotes the marginal distribution on $Q$ of $\pi_{X,X^-,Q}$. The joint distribution $\pi_{Y,X,X^-,Q}$ is known to the encoder and the decoder.

It is further assumed that the encoder and the decoder share some common randomness:
\begin{itemize}
\item a random initial $Q$-state\footnote{For the purposes of this description, we use the term ``$Q$-state'' to denote a node on the $Q$-graph.} $Q_0$, distributed according to $\pi_Q$;
\item a sequence $(U_i)_{i = 1}^n$ of i.i.d. random variables, each $U_i$ being $\text{Unif}[0,1]$-distributed. 
\end{itemize}
At the $i$th time instant, just prior to the transmission of the $i$th symbol by the encoder, both the decoder and the encoder have the output sequence $y^{i-1}$ available to them. From this, and having shared knowledge of a realization $q_0$ of the initial $Q$-state $Q_0$, each of them can compute $q_{i-1} = \Phi(q_0,y^{i-1})$.

The assumption of shared randomness simplifies much of our analysis. It should be noted that by standard averaging arguments, the shared knowledge can be ``de-randomized'', in the sense that there exists a deterministic instantiation of $Q_0$ and $(U_i)$ for which our probability of error analysis will remain valid.

Aside from the shared randomness, the encoder alone has access to an i.i.d.\ sequence $(V_i)_{i=1}^n$, with $V_i \sim \text{Unif}[0,1]$ for all $i$. We next describe the main elements of our PMS.


\subsubsection{Messages and message intervals}
At the outset, there is a set of messages $\cM = \{1,2,\ldots,2^{nR}\}$, where $n$ is a sufficiently large positive integer and $R < I(X;Y \mid Q)$.\footnote{For ease of description, $2^{nR}$ is assumed to be an integer; we may otherwise take the number of messages to be $\lceil 2^{nR} \rceil$.}  A message $M \in \cM$ is selected uniformly at random, and transmitted using the PMS scheme. Our aim is to show that the probability of error $P_e^{(n)}$ at the decoder goes to $0$ exponentially in $n$.

Message intervals are central to the operation of the PMS scheme. At each time instant $i \ge 1$, again just prior to the transmission of the $i$th symbol by the encoder, the decoder and the encoder can compute a common set, $\cJ_i$, of message intervals, which form a partition of $[0,1)$ into disjoint sub-intervals of varying lengths. The construction is recursive, with $\cJ_1$ being computed from the initial knowledge of $q_0$, and for $i \ge 1$, $\cJ_{i+1}$ being computed from $\cJ_i$ based on the additional knowledge of $y_i$ and $u_i$. Here, $u_i$ denotes a realization of the shared randomness $U_i$. We will describe the recursive construction in Section~\ref{sec:recursion}, for now just noting that each message interval in $\cJ_{i+1}$ is obtained from a unique ``parent'' in $\cJ_i$.

Each message interval in $\cJ_i$ is indexed by an integer between $1$ and $|\cJ_i|$, with $\i(\sfJ)$ denoting the index of the message interval $\sfJ$. Additionally, each message interval is associated with a message from $\cM$, the association being specified by a surjective mapping $\mu_i: \cJ_i \to \cM$. Thus, for each $\sfJ \in \cJ_i$, $\mu_i(\sfJ)$ is the message associated with $\sfJ$. One of the message intervals $\sfJ \in \cJ_i$ with $\mu_i(\sfJ) = m$, where $m \in \cM$ is the actual message to be transmitted, is designated as the \emph{true message interval}, and is denoted by $J_i$. The identity of $J_i$ is \emph{a priori} known only to the encoder. This is the message interval that determines the symbol to be transmitted at time $i$.

Each $\sfJ \in \cJ_i$ has a ``history bit'', denoted by $x_i^-(\sfJ)$. This is the bit that the encoder would have transmitted at time $i-1$ if it were the case that the parent of $\sfJ$ was the true message interval $J_{i-1}$.

A message interval $\sfJ \in \cJ_i$ can be uniquely identified by its left end-point $t_i(\sfJ)$ and its length $s_i(\sfJ)$, so that $\sfJ = [t_i(\sfJ), t_i(\sfJ) + s_i(\sfJ))$. The length of $\sfJ$ equals the posterior probability that $\sfJ$ is the true message interval, given $y^{i-1}$ and $u^{i-1}$, i.e.,
\begin{equation}
s_i(\sfJ) = \Pr[J_i = \sfJ \mid Y^{i-1} = y^{i-1}, U^{i-1} = u^{i-1}].
\label{eq:si}
\end{equation}
In section \ref{sec:recursion}, it will be shown that the lengths $s_i(\sfJ)$, $\sfJ \in \cJ_i$, can be computed recursively as a simple function of the lengths of the parent intervals in $\cJ_{i-1}$.

The left end-points are then computed as
\begin{equation}
t_i(\sfJ) = \pi_{X^- | Q}(0 | q_{i-1}) {\mathbbm 1}_{\{x_i^-(\sfJ) = 1\}} + \sum_{\substack{\sfJ'  : \ \i(\sfJ') < \i(\sfJ), \\ \ \ \ \ \ \ x_i^-(\sfJ') = x_i^-(\sfJ)}} s_i(\sfJ').
\label{eq:ti}
\end{equation}
We will show later (see Lemma~\ref{lem:3}) that
$$
\sum_{\sfJ \in \cJ_i: x_i^-(\sfJ) = 0} s_i(\sfJ) = \pi_{X^- | Q}(0 | q_{i-1}),
$$
so that the positioning of the left end-points in \eqref{eq:ti} implies that the message intervals with history bit $0$ (respectively, $1$) form a partition of $[0,\pi_{X^- | Q}(0 | q_{i-1}))$ (respectively, $[\pi_{X^- | Q}(0 | q_{i-1}),1)$). The positioning of the message intervals is illustrated in Fig. \ref{fig:intervals}.
\begin{figure}[t]
\centering
        \psfrag{Q}[][][1]{$\i(\sfJ)=1$}
        \psfrag{E}[][][1]{$\!\!\!\!\!\!\!\!\!\! x_i^-(\cdot)=0$}
        \psfrag{F}[][][1]{\ \ \ \ \ \ $x_i^-(\cdot)=1$}
        \psfrag{I}[][][1]{}
        \psfrag{H}[][][1]{$\pi_{X^-|Q}(0|q_{i-1})$}
        \psfrag{S}[][][1]{$s_i(\sfJ)$}
        \psfrag{O}[][][1]{$1$}
        \psfrag{Z}[][][1]{$t_i(\sfJ)=0$}
        \centerline{\includegraphics[scale=1]{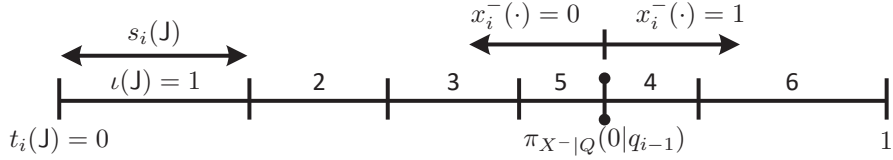}}
\caption{Illustration of the message intervals. Note that the lengths of all messages intervals with $x^-_i=0$ sum up to $\pi_{X^-|Q}(0|q)$.}
\label{fig:intervals}
\end{figure}

In summary, a message interval $\sfJ \in \cJ_i$ stores five pieces of data: $\i(\sfJ), \mu_i(\sfJ)$, $s_i(\sfJ)$, $t_i(\sfJ)$ and $x_i^-(\sfJ)$. As we will see in Section~\ref{sec:recursion}, these are computable at both the encoder and the decoder from their common knowledge at time $i$.

\medskip

\subsubsection{Encoder operation at time $i$}
To describe the symbol transmitted by the encoder at time $i$, we need to introduce the labeling functions $\cL_q: [0,1) \to \{0,1\}$, defined for each $q \in \cQ$. The labeling $\cL_q$ assigns the label `$1$' to the interval $[0,\pi_{X,X^-|Q}(1,0|q))$, and `$0$' to the interval $[\pi_{X,X^-|Q}(1,0|q),1)$. In other words,
\begin{equation}
\cL_q(x) =
\begin{cases}
1 & \text{ if } 0 \le x < \pi_{X,X^-|Q}(1,0|q) \\
0 & \text{ if } \pi_{X,X^-|Q}(1,0|q) \le x < 1.
\end{cases}
\label{Lq}
\end{equation}
The labeling $\cL_q$ is depicted in Fig.~\ref{fig:scheme}(c).

At time $i$, the encoder knows the true message interval $J_i$ in $\cJ_i$. Let $t_i$ and $s_i$ denote its left end-point and length, respectively, and let $x_i^-$ be its history bit. If $x_i^- = 1$, then $J_i$ is contained in $[\pi_{X^- | Q}(0 | q_{i-1}),1)$, which is a subset of $\cL_{q_{i-1}}^{-1}(0) = [\pi_{X,X^-|Q}(1,0|q_{i-1}),1)$. In this case, the encoder transmits $x_i = \cL_{q_{i-1}}(J_i) = 0$, in keeping with the $(1,\infty)$-RLL constraint.

On the other hand, if $x_i^- = 0$, the encoder transmits the bit $x_i = \cL_{q_{i-1}}(w_i)$ with
$$
w_i = t_i + u_i(q_{i-1}) + s_i \cdot v_i \!\! \mod \pi_{X^-|Q}(0|q_{i-1}),
$$
where $u_i(q_{i-1}) := u_i\cdot\pi_{X^-|Q}(0|q_{i-1})$, and $v_i$ denotes a realization of the encoder's private randomness $V_i \sim \text{Unif}[0,1]$. In other words, the encoder picks $w_i$ uniformly at random from the interval $J_i + u_i(q_{i-1}) \mod \pi_{X^-|Q}(0|q_{i-1})$, which is obtained by cyclically shifting the message interval $J_i$ by the amount $u_i(q_{i-1})$, within $[0,\pi_{X^-|Q}(0|q_{i-1}))$.

\begin{remark}
The rationale behind the cyclic shifting is the following: the random variable $W_i = t_i + U_i(q_{i-1}) + s_i \cdot V_i \!\! \mod \pi_{X^-|Q}(0|q_{i-1})$ is uniformly distributed over $[0,\pi_{X^-|Q}(0|q_{i-1}))$. Hence, $X_i = \cL_{q_{i-1}}(W_i)$ is equal to $1$ with probability $\frac{\pi_{X,X^-|Q}(1,0|q_{i-1})}{\pi_{X^-|Q}(0|q_{i-1})} = p^*_{X|X^-,Q}(1|0,q_{i-1})$, and is equal to $0$ with probability $p^*_{X|X^-,Q}(0|0,q_{i-1})$. Thus, the cyclic shifting ensures that the conditional distribution of $X_i$ (given that the previously transmitted bit was $x_i^-$ and the $Q$-state just prior to the transmission of $X_i$ is $q_{i-1}$) matches the optimal input distribution $p^*_{X|X^-,Q}(\cdot | x_i^-,q_{i-1})$.
\end{remark}
\medskip

In the analysis of the PMS schemes of \cite{shayevitz_posterior_mathcing} and \cite{Li_elgamal_matching}, most of the effort goes in showing that the length of the true message interval gets arbitrarily close to $1$, with high probability, as the number of transmissions $n$ goes to $\infty$. In our case, however, the recursive construction of $\cJ_{i+1}$ from $\cJ_i$ involves a key message-interval splitting operation, which prevents the lengths of message intervals from getting too large. Nonetheless, we can show that the length of the true message interval $s_i(J_i)$ either exceeds
\begin{align}\label{eq:Smin}
S_{\min} := \min_{x,q: \ \pi_{X,X^-|Q}(x,0|q) > 0} \pi_{X,X^-|Q}(x,0|q)
\end{align}
at some time $i$, or as the time index $i$ gets close to $n$, $s_i(J_i)$ exceeds a certain threshold $\xi > 0$ with high probability. The following theorem, proved in Section~\ref{sec:proof_scheme}, gives a rigorous statement of this fact.

\begin{theorem}\label{theorem:part_1}
Given $R<I(X;Y|Q)$, there exists $\xi > 0$ (which may depend on $R$) for which the following holds: for any sequence of non-negative integers $(\zeta_n)_{n \ge 1}$ growing as $o(n)$,\footnote{This means that $\zeta_n/n \to 0$ as $n \to \infty$.}
a PMS initiated with $2^{nR}$ messages has
$$\Pr\bigl[s_{n-\zeta_n}(J_{n-\zeta_n})<\xi \mid s_i(J_i) \le S_{\min}, \ i = 1,2,\ldots,n-\zeta_n-1 \bigr] \longrightarrow 0$$
as $n \to \infty$.
\end{theorem}

Recall that the decoder does not know which of the intervals in $\cJ_i$ is the true message interval, but it is able to compute the lengths $s_i(\sfJ)$ for all $\sfJ \in \cJ_i$. The theorem above allows the decoder to create a relatively short list of potential true message intervals.

\subsubsection{Decoder decision}\label{sec:PMS_decoder}

Let $\xi^* = \min\{\xi, S_{\min}\}$, where $\xi$ is as in Theorem~\ref{theorem:part_1} above, and $S_{\min}$ is as defined in \eqref{eq:Smin}. The decoder halts operations at time $T = n - \lceil\sqrt{n}\rceil$, and outputs the list of messages
\begin{equation}
\widehat{\cM} = \{\mu_i(\sfJ): \sfJ \in \cJ_i, \ s_i(\sfJ) \ge \xi^*, 1 \le  i \le T\}. \label{eq:hatcM}
\end{equation}
In words, this is the set of messages associated with message intervals whose lengths exceed either $S_{\min}$ or $\xi$ at any point during the operation of the PMS. This signals the end of Phase~I of our coding scheme.

 Note that since the lengths of all message intervals at time $i$ must sum to $1$, there can be at most $\lfloor 1/\xi^* \rfloor$ messages contributed by $\cJ_i$ to $\widehat{\cM}$, for each $i \in \{1,2,\ldots,T\}$. Thus, we have $|\widehat{\cM}| \le \frac{1}{\xi^*} T \le \frac{1}{\xi^*} n$.

\subsection{Phase~II: Clean-Up} \label{sec:phase2}

The message list generated at the end of Phase~I serves as the input for Phase~II, a complementary coding scheme to determine the correct message within $\widehat{\cM}$. The rate of the coding scheme in this phase can be made to go to zero, since it only has to distinguish between $O(n)$ many messages in ${\cM}$. Each message in $\widehat{\cM}$ is represented using $k_n = \lceil\log_2(n/\xi^*)\rceil$ bits, $b_1b_2\dots b_{k_n}$, which are to be transmitted successively.

\subsubsection{Encoding}
Given a string of $k$ bits, $b_1,\dots b_{k_n}$, the encoder transmits each bit $b_i$ using a length-$L_n$ codeword $b_i0b_i0\ldots b_i0$, where $L_n$ is a suitably chosen even number. Thus, $k_nL_n$ channel uses are required to transmit these $k_n$ bits, but as we will see, $L_n$ will be chosen so that $k_nL_n = o(n)$, so that this is an asymptotically vanishing fraction of the overall number of channel uses, $n$. Note that the encoder does not make use of feedback.

\subsubsection{Decoding}
Based on the sequence of $L_n/2$ channel outputs received in response to the $L_n/2$ repetitions of the bit $b_i$, the decoder declares $\hat{b}_i=0$ if the output sequence lies in the typical set $\mathcal{T}^{(L_n/2)}_{\epsilon} (p_{Y|X=0})$ for a well-chosen $\epsilon > 0$, and declares $\hat{b}_i=1$ otherwise. The decoder's estimate of the transmitted message $M$ is the message $\hat{M}$ represented by $\hat{b}_1\hat{b}_2\dots \hat{b}_{k_n}$.

The next lemma shows that the probability of decoding error for one message bit $b$ can be made arbitrarily small, by choosing $L_n$ as a suitably slowly growing function of $n$. As usual, $\hat{b}$ in the statement of the lemma refers to the decoder's estimate of $b$.
\begin{lemma}\label{lemma:part_2}
If the capacity of the BIBO channel (without feedback or input constraints) is nonzero, then there exists a constant $C_0 > 0$ such that for any $L_n \ge C_0 \log k_n$, we have $\Pr[\hat{b} \neq b] \leq \frac{1}{k_n^2}$.
\end{lemma}
\begin{proof}
We use a standard typicality argument based on the fact that if capacity is non-zero, then
\begin{align*}
   D(p_{Y|X=0}||p_{Y|X=1})&\neq 0,
\end{align*}
and so there exists a sequence $\epsilon(\ell)>0$ such that, for all sufficiently large $\ell$, $\mathcal{T}^{(\ell)}_{\epsilon(\ell)} (p_{Y|X=0})\bigcap\mathcal{T}^{(\ell)}_{\epsilon(\ell)} (p_{Y|X=1})=\emptyset$ and both typical sets are nonempty.

Without loss of generality, assume that the transmitted bit is $b=1$, so that the sequence of channel outputs in response to the sequence of transmitted $1$s is i.i.d.\ $\sim p_{Y|X=1}$. By standard arguments, it can be shown that the probability that a length-$L/2$ output sequence is in $\mathcal{T}^{(L/2)}_{\epsilon'} (p_{Y|X=1})$ (for some $0<\epsilon'<\epsilon(L/2)$) goes to $1$ exponentially quickly in $L$, while the probability that the output sequence is in $\mathcal{T}^{(L/2)}_{\epsilon'} (p_{Y|X=0})$ decays to $0$ exponentially quickly in $L$. Therefore, $\Pr[\hat{b} \neq b]$ can be made smaller than $\frac{1}{k_n^2}$ by choosing $L_n = C_0 \log k_n$ for a sufficiently large positive constant $C_0$.
\end{proof}

\subsection{Combining Phases~I and II} \label{sec:combine}

We now describe how Phases~I and II are combined to obtain a capacity-achieving coding scheme.
Fix an $R < I(X;Y | Q)$ and let $\xi$ be as in Theorem~\ref{theorem:part_1}, which in turn determines $\xi^* = \min\{\xi, S_{\min}\}$. We will apply Theorem~\ref{theorem:part_1} with $\zeta_n = \lceil\sqrt{n}\rceil$. Let $k_n = \lceil \log_2(n/\xi^*)\rceil$, and $L_n = \zeta_n/k_n$ be the parameters of the coding scheme in Phase~II. Note that $L_n \ge C_0 \log k_n$ for all sufficiently large $n$, where $C_0$ is the constant in the statement of Lemma~\ref{lemma:part_2}.

%

We will run the PMS of Phase~I on a message set $\cM$ of size $2^{nR}$. A message $M \sim \text{Unif}(\cM)$ is transmitted using $n$ uses of the channel as follows: The PMS of Phase~I is executed until time $T = n-\zeta_n$ (i.e., $T=n-\lceil\sqrt n\rceil$), at which time a list $\widehat{M}$ as in \eqref{eq:hatcM} is produced. We then execute Phase~II for the remaining $k_nL_n = \zeta_n$ channel uses, at the end of which the decoder produces an estimate $\hat{M}$ of $M$. The coding rate of the overall scheme is $\frac{1}{n}\,\log_2 |\cM| = R$.

To assess the probability of error $P_e^{(n)}$, we observe that
\begin{align}
P_e^{(n)} & \ = \ \Pr[M \ne \hat{M}] \notag \\
& \ \le \ \Pr[M \notin \widehat{\cM}] +  \Pr[M \neq \hat{M} \mid M \in \widehat{\cM}] \notag \\
& \ \le \ \Pr[M \notin \widehat{\cM}] + \sum_{i=1}^{k_n} \Pr[\hat{b}_i \ne b_i] \notag \\
&\ \le \ \Pr[M \notin \widehat{\cM}] + \frac{1}{k_n}, \label{ineq:Pe}
\end{align}
the last inequality above being valid for all $n$ large enough that $L_n \ge C_0 \log k_n$, so that the conclusion of Lemma~\ref{lemma:part_2} holds. The probability that $M \notin \widehat{\cM}$ can be bounded as follows:
\begin{align*}
\Pr[M \notin \widehat{\cM}]  & \ \le \ \Pr[s_i(J_i) \le S_{\min} \text{ for } i = 1,2,\ldots,n-\zeta_n-1, \text{ and }  s_{n-\zeta_n}(J_{n-\zeta_n}) < \xi] \\
& \ \le \Pr\bigl[s_{n-\zeta_n}(J_{n-\zeta_n}) < \xi \mid s_i(J_i) \le S_{\min} \text{ for } i = 1,2,\ldots,n-\zeta_n-1 \bigr],
\end{align*}
which, by Theorem~\ref{theorem:part_1}, goes to $0$ as $n \to \infty$. Hence, by \eqref{ineq:Pe}, we also have $P_e^{(n)} \to 0$ as $n \to \infty$.

We have thus shown that any rate $R < I(X;Y|Q)$ is $(1,\infty)$-achievable. Recall from Theorem~\ref{theorem:optimal_inputs} that $I(X;Y|Q) = C^{\mathrm{fb}}(\alpha,\beta)$. Thus, we have proved the main result of this section, stated below.

\begin{theorem}\label{theorem:scheme_optimal}
For an input-constrained BIBO channel, any rate $R < C^{\mathrm{fb}}(\alpha,\beta)$ is $(1,\infty)$-achievable using a coding scheme that combines Phases~I and II.
\end{theorem}

It remains to tie a couple of loose ends in the description of our coding scheme, namely, the recursive construction of message intervals, and a proof of Theorem~\ref{theorem:part_1}. The former is presented in the subsection below, while the latter is given in Section~\ref{sec:proof_scheme}.

\subsection{Recursive Construction of Message Intervals}\label{sec:recursion}

\subsubsection{Initialization --- Construction of $\cJ_1$} $\cJ_1$ consists of the $2^{nR}$ intervals $\sfJ^{(j)} = [(j-1)2^{-nR},j2^{-nR})$ of equal length, indexed by $j \in \{1,2,\ldots,2^{nR}\}$, in 1-1 correspondence with the $2^{nR}$ messages in $\cM$. The length of each message interval is $s_1 := 2^{-nR}$. For the $j$th message interval $\sfJ^{(j)}$, the index $\i(\sfJ^{(j)})$ and associated message $\mu(\sfJ^{(j)})$ are both set to be $j$. The history bit of the $\sfJ^{(j)}$ is set to be
$$
x_1^-(\sfJ^{(j)}) = \begin{cases} 0 & \text{ if }  (j-1) 2^{-nR} \le \pi_{X^- | Q} (0 | q_0) \\ 1 & \text{ otherwise,} \end{cases}
$$
where $q_0$ is the initial $Q$-state known to both encoder and decoder.

We will, for simplicity of description, assume that $n$ is chosen so that the message intervals in $\cJ_1$ with history bit equal to $0$ form a partition of $[0,\pi_{X^- | Q} (0 | q_0))$.\footnote{If this is not the case, we can get this to happen by splitting into two the message interval that straddles the boundary point $\pi_{X^- | Q} (0 | q_0)$, as described in the next subsection.}  For a uniformly random message $M \in \cM$, the true message interval is $J_1 = \sfJ^{(M)}$, and is \emph{a priori} known only to the encoder. Note that with $X_0 = x_1^-(J_1)$, the pair $(X_0,Q_0)$ has probability distribution $\pi_{X^- | Q} \pi_Q = \pi_{X^-,Q}$.

\medskip

\subsubsection{Recursion --- Construction of $\cJ_{i+1}$ from $\cJ_i$} Recall that, for $i \ge 1$, the encoder and decoder can compute the set $\cJ_i$ based upon their common knowledge of $q_0$, $y^{i-1}$ and $u^{i-1}$. 
After determining $\cJ_i$, they make use of their shared knowledge of $u_i$ to compute a new partition, $\tcJ_i$, of $[0,1)$ into message intervals. This is an intermediate step in the construction of $\cJ_{i+1}$ from $\cJ_i$. Message intervals $\sfJ \in \tcJ_i$ also store five pieces of data: an index $\ti(\sfJ)$, a message $\tmu_i(\sfJ) \in \cM$, a length $\ts_i(\sfJ)$, a left end-point $\tildet_i(\sfJ)$, and a bit $x_i(\sfJ)$. These are explained as part of the description to follow.

Each message interval $\sfJ \in \cJ_i$ is either retained as is in $\tcJ_i$, or is split into two intervals in $\tcJ_i$ using a procedure to be described shortly. Any message interval $\sfJ \in \cJ_i$ with history bit $x_i^-(\sfJ) = 1$ is retained as is in $\tcJ_i$. It retains its index, message, length and left end-point: $\ti(\sfJ) = \i(\sfJ)$, $\tmu_i(\sfJ) = \mu_i(\sfJ)$, $\ts_i(\sfJ) = s_i(\sfJ)$, and $\tildet_i(\sfJ) = t_i(\sfJ)$. The bit $x_i(\sfJ)$ is set to be equal to $0$, in keeping with the input constraint.

\paragraph{Cyclic shifting and message-interval splitting}
To describe what happens to those $\sfJ \in \cJ_i$ with history bit $x_i^-(\sfJ) = 0$, we first recall that such message intervals form a partition of $[0,\pi_{X^-|Q}(0|q_{i-1}))$ --- see Lemma~\ref{lem:3}.

The fate of message intervals $\sfJ \in \cJ_i$ with history bit $x_i^-(\sfJ) = 0$ is determined by the labeling $\cL_{q_{i-1}}$, as defined in \eqref{Lq}. (Recall that $q_{i-1}$ can be determined at both the encoder and the decoder from their common knowledge of $q_0$ and $y^{i-1}$.) Recall from the description of the encoder operation in Section~\ref{sec:phase1} that when the history bit $x_i^-$ for the true message interval $J_i$ is equal to $0$, then the encoder determines the next bit to be transmitted as $\cL_{q_{i-1}}(w_i)$, where $w_i$ is chosen uniformly at random from the interval obtained by cyclically shifting $J_i$ within the interval $[0,\pi_{X^-|Q}(0|q_{i-1}))$ by $u_i(q_{i-1}) := u_i \cdot \pi_{X^-|Q}(0|q_{i-1})$. Since the decoder does not know the true $J_i$, it attempts to keep up with the encoder by applying the cyclic shifting operation to each $\sfJ \in \cJ_i$. This results in a cyclically-shifted interval $\sfJ^{u_i}$, for each $\sfJ \in \cJ_i$, with left end-point $t_i(\sfJ^{u_i}) := t_i(\sfJ) + u_i(q_{i-1}) \!\! \mod \pi_{X^-|Q}(0|q_{i-1})$, and right end-point $r_i(\sfJ^{u_i}) := t_i(\sfJ) + s_i(\sfJ) + u_i(q_{i-1}) \!\! \mod \pi_{X^-|Q}(0|q_{i-1})$. If $\sfJ^{u_i} \subseteq \cL_{q_{i-1}}^{-1}(\mathsf{b})$ for some $\mathsf{b} \in \{0,1\}$, then $\sfJ^{u_i}$ is included in $\tcJ_i$, with
$$
\ti(\sfJ^{u_i}) = \i(\sfJ), \ \ \tmu_i(\sfJ^{u_i}) = \mu_i(\sfJ), \ \ \ts_i(\sfJ^{u_i}) = s_i(\sfJ), \ \ \tildet_i(\sfJ^{u_i}) = t_i(\sfJ),
$$
and
$$
x_i(\sfJ^{u_i}) = \begin{cases}
0 & \text{ if } \sfJ^{u_i} \subseteq \cL_{q_{i-1}}^{-1}(0) \\
1 & \text{ if } \sfJ^{u_i} \subseteq \cL_{q_{i-1}}^{-1}(1).
\end{cases}
$$
Note that $x_i(\sfJ^{u_i})$ would have been the bit transmitted by the encoder at time $i$, had $\sfJ$ been the true message interval.

A problem arises when $\sfJ^{u_i} \not\subseteq \cL_{q_{i-1}}^{-1}(\mathsf{b})$ for any $\mathsf{b} \in \{0,1\}$, as it then straddles at least one of the boundary points, $0$ and $\pi_{X,X^-|Q}(1,0|q_{i-1})$, of the labeling $\cL_{q_{i-1}}$. This means that the value of the bit $x_i$ transmitted by the encoder at time $i$, had this $\sfJ$ been the true message interval, is determined by the precise location of the random point $w_i$ within $\sfJ^{u_i}$. While this is not a problem for the encoder, it creates an issue for the decoder as it will no longer know what to assign as the bit $x_i(\sfJ^{u_i})$. We deal with this by \emph{splitting} $\sfJ^{u_i}$ into two or three parts. To describe this, we first observe that if the length of $\sfJ^{u_i}$ (which is the same as $s_i(\sfJ$)) is at most $S_{\min}$ (Eq. \eqref{eq:Smin}), then any $\sfJ^{u_i} \not\subseteq \cL_{q_{i-1}}^{-1}(\mathsf{b})$, $\mathsf{b} = 0,1$, can straddle exactly one of the boundary points of $\cL_{q_{i-1}}$.

If $\sfJ^{u_i}$ straddles only the boundary point $0$, then we split $\sfJ^{u_i}$ into two intervals $\sfJ' = [t_i(\sfJ^{u_i}),\pi_{X^-|Q}(0|q_{i-1}))$ and $\sfJ'' = [0,r_i(\sfJ^{u_i}))$ to be included in $\tcJ_i$, with $\tmu_i(\sfJ') = \tmu_i(\sfJ'') = \mu_i(\sfJ)$. The left end-points $\tildet_i$ and lengths $\ts_i$ of $\sfJ'$ and $\sfJ''$ are self-evident. Note that $\sfJ'$ sits entirely within $\cL_{q_{i-1}}^{-1}(0)$, so that we set $x_i(\sfJ') = 0$. By similar reasoning, we set $x_i(\sfJ'') = 1$. Finally, we set $\ti(\sfJ') = \i(\sfJ)$, while $\sfJ''$ is assigned a brand new index: $\ti(\sfJ'')$ is set to be equal to the least positive integer that has not yet been assigned as an index to any message interval in $\tcJ_i$.

\begin{figure}[t]
\centering
        \psfrag{A}[][][1]{Message}
        \psfrag{B}[][][1]{Random}
        \psfrag{C}[][][1]{Labelling}
        \psfrag{D}[][][1]{Intervals}
        \psfrag{E}[][][1]{$\!\!\!\!\!\!\!\!\!\! x_i^-(\cdot)=0$}
        \psfrag{F}[][][1]{\ \ \ \ \ \ $x_i^-(\cdot)=1$}
        \psfrag{G}[][][1]{$'1'$}
        \psfrag{H}[][][1]{$'0'$}
        \psfrag{I}[][][1]{$'0'$}
        \psfrag{J}[][][0.9]{$\pi_{X,X^-|Q}(1,0|q_{i-1})$}
        \psfrag{K}[][][0.9]{$\pi_{X^-|Q}(0|q_{i-1})$}
        \psfrag{L}[][][1]{$s_i(\sfJ)$}
        \psfrag{M}[][][1]{intervals}
        \psfrag{N}[][][1]{cyclic shift}
        \psfrag{O}[][][1]{$\cL_{q_{i-1}}$}
        \psfrag{P}[][][1]{after split}
        \psfrag{Q}[][][0.8]{$\i(\sfJ)=1$}
        \psfrag{Z}[][][0.8]{$\tilde{s}_i(\sfJ')$}
        \psfrag{S}[][][0.8]{$t_i(\sfJ)$}
        \psfrag{T}[][][0.8]{$t_i(\sfJ^{u_i})$}
        \psfrag{V}[][][0.75]{$\tilde{\i}(\sfJ')=1$}
        \psfrag{R}[][][0.75]{$\tilde{\i}(\sfJ'')=7$}
        \centerline{\includegraphics[scale=1]{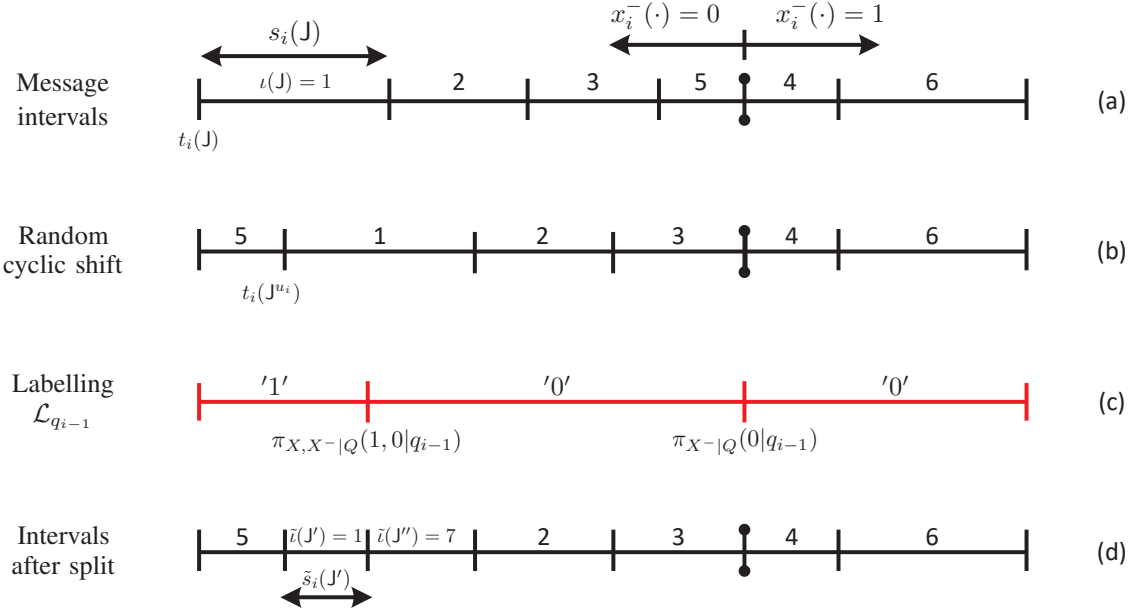}}
\caption{Illustration of the cyclic shifting and splitting of message intervals. The lengths of message intervals are determined by their posterior probabilities, and the message intervals are positioned in ascending order of their indices, based on their history bits. Each message interval $\sfJ$ with $x_i^-(\sfJ)=0$ is cyclically shifted within the interval $[0,\pi_{X^-|Q}(0|q_{i-1}))$ by adding $u_i(q_{i-1})$, and its new left end-point is $t_i(\sfJ^{u_i})$. The labeling $\cL_{q_{i-1}}$ implies that the message interval indexed with $\i(\sfJ)=1$ crosses the boundary point $\pi_{X,X^-|Q}(1,0|q_{i-1})$. Thus, this message interval is split into two intervals $\sfJ'$ and $\sfJ''$, belonging to $\tilde{\cJ}_{i}$, and indexed with $\tilde{\i}(\sfJ')=1$ and $\tilde{\i}(\sfJ'')=7$.}
\label{fig:scheme}
\end{figure}

If $\sfJ^{u_i}$ straddles only the boundary point $\pi_{X,X^-|Q}(1,0|q_{i-1})$, then we split $\sfJ^{u_i}$ into the two intervals $\sfJ' = [t_i(\sfJ^{u_i}),\pi_{X,X^-|Q}(1,0|q_{i-1}))$ and $\sfJ'' = (\pi_{X,X^-|Q}(1,0|q_{i-1}), r_i(\sfJ^{u_i})$ to be included in $\tcJ_i$. For the new message intervals $\sfJ'$ and $\sfJ''$, we set $\tmu_i$ and $\ti$ as above, while $x_i(\sfJ') = 1$ and $x_i(\sfJ'') = 0$. Fig.~\ref{fig:scheme} illustrates the cyclic shifting and splitting operations on message intervals.

Finally, if $s_i(\sfJ) > S_{\min}$,\footnote{While this case needs to be included for a complete description of the recursive construction of $\cJ_{i+1}$ from $\cJ_i$, it will play no role in the analysis of the PMS in Section~\ref{sec:proof_scheme}. The analysis there is carried out under the simplifying assumption that the length of the true message interval never exceeds $S_{\min}$.}  it is possible for $\sfJ^{u_i}$ to stretch across both $0$ and $\pi_{X,X^-|Q}(1,0|q_{i-1})$. In this case, we split $\sfJ^{u_i}$ into three intervals $\sfJ'$, $\sfJ''$, $\sfJ'''$ such that each of these intervals lies entirely within one of $\cL_{q_{i-1}}^{-1}(0)$ and $\cL_{q_{i-1}}^{-1}(1)$. Then, for each of these new intervals, we set the bit $x_i$ to be the $\mathsf{b} \in \{0,1\}$ for which the interval lies in $\cL_{q_{i-1}}^{-1}(\mathsf{b})$. The left end-points $\tildet_i$ and lengths $\ts_i$ are determined in the obvious manner. Finally, $\tmu_i(\sfJ') = \tmu_i(\sfJ'') = \tmu_i(\sfJ''') = \mu_i(\sfJ)$, and $\ti(\sfJ') = \i(\sfJ)$, while $\sfJ''$ and $\sfJ'''$ get brand new indices.

Thus, each $\sfJ \in \cJ_i$ gives rise to either one interval (no splitting) or two to three message intervals (after splitting) in $\tcJ_i$; we will refer to the interval(s) in $\tcJ_i$ as the \emph{child(ren)} of $\sfJ \in \cJ_i$, and to $\sfJ$ as their \emph{parent}. Note that splitting affects only those message intervals in $\cJ_i$ that, when cyclically shifted by $u_i(q_{i-1})$, straddle the boundary points $0$ and $\pi_{X,X^-|Q}(1,0|q_{i-1})$. Thus, splitting causes the number of intervals to increase by at most two: $|\cJ_i| \le |\tcJ_i| \le |\cJ_i|+2$.

The splitting of message intervals lends itself to a simple alternative description of the encoding operation described earlier in Section~\ref{sec:phase1}. Suppose that $J_i = \sfJ \in \cJ_i$. If $\sfJ^{u_i}$ has exactly one child $\tilde{\sfJ} \in \tcJ_i$, then $x_i = x_i(\tilde{\sfJ})$. If $\sfJ$ has two children $\sfJ', \sfJ''$ in $\tcJ_i$, then the encoder uses $v_i$ to determine the transmitted bit $x_i$:
\begin{equation}
x_i = \begin{cases}
x_i(\sfJ') & \text{ if } 0 \le v_i < \frac{\ts_i(\sfJ')}{s_i(\sfJ)} \\
x_i(\sfJ'') & \text{ if } \frac{\ts_i(\sfJ')}{s_i(\sfJ)} \le v_i < 1.
\end{cases}
\label{eq:xi}
\end{equation}
The case when $\sfJ$ has three children is handled analogously:
\begin{equation}
x_i = \begin{cases}
x_i(\sfJ') & \text{ if } 0 \le v_i < \frac{\ts_i(\sfJ')}{s_i(\sfJ)} \\
x_i(\sfJ'') & \text{ if }  \frac{\ts_i(\sfJ')}{s_i(\sfJ)} \le v_i < \frac{\ts_i(\sfJ')+\ts_i(\sfJ'')}{s_i(\sfJ)} \\
x_i(\sfJ''') & \text{ if } \frac{\ts_i(\sfJ') + \ts_i(\sfJ'')}{s_i(\sfJ)} \le v_i < 1.
\end{cases}
\label{eq:xi3}
\end{equation}

\medskip

\paragraph{Construction of $\cJ_{i+1}$} The message intervals in $\cJ_{i+1}$ are in 1-1 correspondence with those in the set $\tcJ_i$. Specifically, for each $\tilde{\sfJ} \in \tcJ_i$, we introduce a message interval $\sfJ \in \cJ_{i+1}$, with $\i(\sfJ) = \ti(\tilde{\sfJ})$, $\mu_{i+1}(\sfJ) = \tmu_i(\tilde{\sfJ})$ and $x_{i+1}^-(\sfJ) = x_i(\tilde{\sfJ})$. We will refer to the message interval $\sfJ$ as the \emph{image} of $\tilde{\sfJ}$. The new interval lengths
\begin{equation}
s_{i+1}(\sfJ) = \Pr[J_{i+1} = \sfJ \mid Y^i = y^i, U^i = u^i], \ \ \forall \, \sfJ \in \cJ_{i+1},
\label{eq:si+1}
\end{equation}
can be determined from the set of interval lengths in $\tcJ_i$, as we will describe shortly. Once the interval lengths $s_{i+1}(\sfJ)$ have been determined, the left end-points can be computed as
\begin{equation}
t_{i+1}(\sfJ) = \pi_{X^- | Q}(0 | q_{i}) {\mathbbm 1}_{\{x_{i+1}^-(\sfJ) = 1\} +}\sum_{\substack{\sfJ'  : \ \i(\sfJ') < \i(\sfJ), \\ \ \ \ \ \ \ x_{i+1}^-(\sfJ') = x_{i+1}^-(\sfJ)}} s_{i+1}(\sfJ').
\label{eq:ti+1}
\end{equation}

To be able to describe how the lengths of intervals in $\cJ_{i+1}$ are computed from the lengths of intervals in $\tcJ_i$, we need to understand how the encoder decides the symbol $x_i$ to be transmitted at time $i$, and how it picks the true message interval $J_{i+1} \in \cJ_{i+1}$.
\medskip

\paragraph{Choice of the true message interval $J_{i+1} \in \cJ_{i+1}$}

Suppose that $J_i = \sfJ \in \cJ_i$. If $\sfJ$ has exactly one child $\tilde{\sfJ} \in \tcJ_i$, then $J_{i+1}$ is taken to be the image of $\tilde{\sfJ}$ in $\cJ_{i+1}$, i.e., the message interval in $\cJ_{i+1}$ that has the same index as $\sfJ$ (and $\tilde{\sfJ}$). On the other hand, if $\sfJ$ has two children $\sfJ', \sfJ''$ in $\tcJ_i$, then $J_{i+1}$ is set to be the image of $\sfJ'$ (respectively, $\sfJ''$) if the transmitted bit $x_i$ in \eqref{eq:xi} equals $x_i(\sfJ')$ (respectively, $x_i(\sfJ'')$). In any case, note that we always have $x_{i+1}^-(J_{i+1}) = x_i$. The case when $\sfJ$ has three children is similarly handled, based on \eqref{eq:xi3}.

\medskip

\subsubsection{Computing $s_{i+1}(\sfJ)$, $\sfJ \in \cJ_{i+1}$, from $\ts_i(\sfJ)$, $\sfJ \in \tcJ_i$}\label{sec:si+1}

In what follows, we use $\sfJ$ to denote a message interval in $\tcJ_i$ as well as its image in $\cJ_{i+1}$. We state a preliminary lemma first.

\begin{lemma} For any $\sfJ \in \cJ_{i+1}$, we have
$\Pr[J_{i+1} = \sfJ \mid Y^{i-1} = y^{i-1}, U^i = u^i] = \ts_i(\sfJ)$.
\label{lem:1}
\end{lemma}

\begin{proof}
Let $\sfJ^{\#} \in \cJ_i$ be the parent of $\sfJ$, so that $J_{i+1} = \sfJ$ implies that $J_i = \sfJ^{\#}$. Then,
\begin{align*}
\Pr[J_{i+1} = \sfJ & \mid Y^{i-1} = y^{i-1}, U^i = u^i]  \\
&= \Pr[J_{i+1} = \sfJ, J_i = \sfJ^{\#}  \mid Y^{i-1} = y^{i-1}, U^i = u^i] \\
&= \Pr[J_i = \sfJ^{\#} \mid Y^{i-1} = y^{i-1}, U^{i-1} = u^{i-1}] \, \Pr[J_{i+1} = \sfJ \mid J_i = \sfJ^{\#}, U_i = u_i] \\
&= s_i(\sfJ^{\#}) \, \frac{\ts_i(\sfJ)}{s_i(\sfJ^{\#})},
\end{align*}
where $\Pr[J_{i+1} = \sfJ \mid J_i = \sfJ^{\#}, U_i = u_i] = \frac{\ts_i(\sfJ)}{s_i(\sfJ^{\#})}$ follows from the way that $J_{i+1}$ is chosen, given $J_i$ and $u_i$.
\end{proof}

\medskip

We are now in a position to derive the means of computing $s_{i+1}(\sfJ)$ from $\ts_i(\sfJ)$.

\begin{lemma} For each $\sfJ \in \cJ_{i+1}$, we have
$$
s_{i+1}(\sfJ) := \Pr[J_{i+1} = \sfJ \mid Y^{i} = y^{i}, U^i = u^i] = \ts_i(\sfJ) \, \frac{p_{Y|X}(y_i \mid x_i(\sfJ))}{\pi_{Y|Q}}(y_i \mid q_{i-1}),
$$
where $q_{i-1} = \Phi(q_0,y^{i-1})$.
\label{lem:2}
\end{lemma}
\begin{proof} We start with
\begin{align*}
\Pr[J_{i+1} = \sfJ \mid Y^{i} = y^{i}, U^i = u^i]
 &= \frac{\Pr[J_{i+1} = \sfJ, Y_i = y_i \mid Y^{i-1} = y^{i-1}, U^i = u^i]}{\Pr[Y_i = y_i \mid Y^{i-1} = y^{i-1}, U^i = u^i]} \\
 &\stackrel{\text{(a)}}{=} \ts_i(\sfJ) \, \frac{ \Pr[Y_i = y_i \mid J_{i+1} = \sfJ, Y^{i-1} = y^{i-1}, U^i = u^i]}{\Pr[Y_i = y_i \mid Y^{i-1} = y^{i-1}, U^i = u^i]} \\
 &\stackrel{\text{(b)}}{=} \ts_i(\sfJ) \, \frac{\Pr[Y_i = y_i \mid X_i = x_i(\sfJ)]}{\Pr[Y_i = y_i \mid Y^{i-1} = y^{i-1}, U^i = u^i]},
\end{align*}
where (a) is by Lemma~\ref{lem:1}, and (b) is due to the fact that if $J_{i+1}=\sfJ$, then the bit transmitted at time instant $i$ must have been $x_{i+1}^-(\sfJ) = x_i(\sfJ)$.

The denominator on the right-hand-side above can be expressed as
\begin{align}
\Pr[Y_i = y_i \mid Y^{i-1} = y^{i-1}, U^i = u^i]
 &= \sum_{\tilde{\sfJ} \in \cJ_{i+1}} \Pr[J_{i+1} = \tilde{\sfJ}, Y_i = y_i \mid Y^{i-1} = y^{i-1}, U^i = u^i]  \notag \\
 &= \sum_{\tilde{\sfJ} \in \tcJ_{i}} \ts_i(\tilde{\sfJ}) \, \Pr[Y_i = y_i \mid X_i = x_i(\tilde{\sfJ})] \label{eq:denom}
\end{align}
Now, the last expression above can be written as
$$\left(\sum_{\tilde{\sfJ} \in \tcJ_i: \ x_i(\tilde{\sfJ}) = 1} \ts_i(\tilde{\sfJ})\right) \Pr[Y_i = y_i \mid X_i = 1]  + \left(\sum_{\tilde{\sfJ} \in \tcJ_i: \ x_i(\tilde{\sfJ}) = 0} \ts_i(\sfJ)\right) \Pr[Y_i = y_i \mid X_i = 0].$$
By construction of $\tcJ_i$ from $\cJ_i$, the message intervals $\tilde{\sfJ} \in \tcJ_i$ with $x_i(\tilde{\sfJ}) = 1$ form a partition of $[0,\pi_{X,X^-|Q}(1|0,q_{i-1}))$, so that
$$
\sum_{\tilde{\sfJ} \in \tcJ_i: \ x_i(\tilde{\sfJ}) = 1} \ts_i(\tilde{\sfJ}) = \pi_{X,X^-|Q}(1,0|q_{i-1})
$$
and
\begin{align*}
\sum_{\tilde{\sfJ} \in \tcJ_i: \ x_i(\tilde{\sfJ}) = 0} \ts_i(\tilde{\sfJ})
&= 1- \sum_{\tilde{\sfJ} \in \tcJ_i: \ x_i(\tilde{\sfJ}) = 1} \ts_i(\tilde{\sfJ}) \\
&= 1 - \pi_{X,X^-|Q}(1,0|q_{i-1}).
\end{align*}
Thus, \eqref{eq:denom} simplifies to
\begin{align*}
\Pr[Y_i = y_i & \mid Y^{i-1} = y^{i-1}, U^i = u^i] \\
& = \ \pi_{X,X^-|Q}(1,0|q_{i-1}) \, p_{Y|X}(y_i|1) + \left(1-\pi_{X,X^-|Q}(1,0|q_{i-1})\right) \, p_{Y|X}(y_i|0) \\
& = \sum_{x,x^-} \pi_{X,X^-|Q}(x,x^-|q_{i-1}) \, p_{Y|X}(y_i|x),
\end{align*}
recalling that $\pi_{X,X^-|Q}(x,x^-|q_{i-1}) = 0$ for $(x,x^-) = (1,1)$.
\end{proof}

\medskip

One simple consequence of Lemma~\ref{lem:2} is that for each $\sfJ \in \cJ_{i+1}$, we have
$$
s_{i+1}(\sfJ) \le s_i(\sfJ^{\#}) \, \max_{x,y,q} \frac{p_{Y|X}(y|x)}{\pi_{Y|Q}(y | q)},
$$
where $\sfJ^{\#} \in \cJ_i$ is the parent of $\sfJ$. Recursively applying this inequality, we obtain \begin{equation}
s_{i+1}(\sfJ) \le 2^{-nR} \, {\left(\max_{x,y,q} \frac{p_{Y|X}(y|x)}{\pi_{Y|Q}(y | q)}\right)}^i,
\label{s_upbnd}
\end{equation}
which is a crude, but useful, upper bound on interval lengths in $\cJ_{i+1}$.

A second consequence of Lemma~\ref{lem:2} is the fact, crucial for our description of cyclic shifting and message-interval splitting, that the message intervals $\sfJ \in \cJ_i$ with history bit $x_i^-(\sfJ) = 0$ form a partition of $[0,\pi_{X^-|Q}(0|q_{i-1}))$. This follows from the next lemma.

\begin{lemma}
For any $i \ge 1$, we have
$$
\sum_{\sfJ \in \cJ_i: \ x_i^-(\sfJ) = 0} s_i(\sfJ) = \pi_{X^- | Q}(0 \mid q_{i-1}).
$$
\label{lem:3}
\end{lemma}
\begin{proof}
The proof is by induction on $i$. By construction, the statement is true for $i = 1$. So, suppose that it holds for some $i \ge 1$. We then consider $\sum_{\sfJ \in \cJ_{i+1}:\  x_{i+1}^-(\sfJ) = 0} s_{i+1}(\sfJ)$. By Lemma~\ref{lem:2}, we have
\begin{align*}
\sum_{\sfJ \in \cJ_{i+1}: \ x_{i+1}^-(\sfJ) = 0} s_{i+1}(\sfJ) &= \sum_{\sfJ \in \tcJ_i: \ x_i(\sfJ) = 0} \ts_i(\sfJ) \, \frac{p_{Y|X}(y_i \mid x_i(\sfJ))}{\pi_{Y|Q}(y_i \mid q_{i-1})} \\
&= \frac{p_{Y|X}(y_i \mid 0)}{\pi_{Y|Q}(y_i \mid q_{i-1})} \, \left(\sum_{\sfJ \in \tcJ_i: \ x_i(\sfJ) = 0} \ts_i(\sfJ)\right) \\
&\stackrel{\text{(a)}}{=} \frac{p_{Y|X}(y_i \mid 0)}{\pi_{Y|Q}(y_i \mid q_{i-1})} \, \left(1-\pi_{X,X^-|Q}(1,0 \mid q_{i-1})\right) \\
&= \frac{\sum_{x^-} p_{Y|X}(y_i \mid 0) \, \pi_{X,X^-|Q}(0,x^- \mid q_{i-1})}{\sum_{x,x^-} p_{Y|X}(y_i|x) \, \pi_{X,X^-|Q}(x,x^- \mid q_{i-1})}  \\
&= \frac{\sum_{x^-} p_{Y|X}(y_i \mid 0) \, p^*_{X|X^-,Q}(0 \mid x^-, q_{i-1}) \, \pi_{X^-|Q}(x^-|q_{i-1})}{\sum_{x,x^-} p_{Y|X}(y_i|x) \, p^*_{X|X^-,Q}(x|x^-,q_{i-1}) \, \pi_{X^-|Q}(x^-|q_{i-1})} \\
&\stackrel{\text{(b)}}{=} \pi_{X^-|Q}(0 \mid q_i),
\end{align*}
where (b) can be verified directly from the known expressions for $p_{Y|X}$, $p^*_{X|X^-,Q}$ and $\pi_{X^-|Q}$. The induction hypothesis has been used in (a) above to validate the construction of $\tcJ_i$ from $\cJ_i$ via message splitting, which yields the fact that $\sum_{\sfJ \in \tcJ_i: x_i(\sfJ) = 0} \ts_i(\sfJ) = 1-\pi_{X,X^-|Q}(1,0 \mid q_{i-1})$.
\end{proof}

\begin{remark}
The PMS construction described in this section and its analysis to come in the next section are also valid for non-optimal input distributions and for a broader class of channels called unifilar finite state channels (with feedback) \cite{PermuterCuffVanRoyWeissman08}, provided that one condition is met. This condition, which may easily go unnoticed within the details of the PMS construction, is that of Lemma~\ref{lem:3}: the sum of the lengths of all the messages intervals with the same history bit $x_i^-$ is a function of the $Q$-state $q_{i-1}$ only. This condition is essentially equivalent to the BCJR-invariant property that was introduced in \cite[Section III.A]{Sabag_UB_IT}. We emphasize that this property is immediately satisfied when an input distribution satisfies the Bellman equation in the corresponding DP problem. However, this property can also be verified directly, as has been done in Lemma \ref{lem:3}.
\end{remark}

\section{PMS Analysis: Proof of Theorem~\ref{theorem:part_1}}\label{sec:proof_scheme}

The statement to be proved concerns the probability of $s_{n-\zeta_n}(J_{n-\zeta_n}) < \xi$, conditioned on the occurrence of the event $s_i(J_i) \le S_{\min}, \ i = 1,2,\ldots,n-\zeta_n-1$. Thus, throughout this section, we assume that $s_i(J_i) \le S_{\min}, \ i = 1,2,\ldots,n-\zeta_n-1$, holds. All probabilities and expectations in this section are implicitly conditioned on this event. This results in a simplified analysis of the PMS in Phase~I, since, under this assumption, the true message interval $J_i$ cannot split into more than two children at any point of the PMS, as is evident from the description of message-interval splitting in Section~\ref{sec:recursion}.

Since this proof is concerned only with the sequence of true message intervals $J_i$, $i = 1,2,3,\ldots$, we will use some simplified notation: $S_{i+1} = s_{i+1}(J_{i+1})$, $\tilde{S}_{i} = \ts_{i}(J_{i+1})$, $X_i = x_i(J_i)$, $X_i^- = x_i^-(J_i) = x_{i-1}(J_{i-1})$.

\subsection{Preliminaries}
Define for $\rho\in[0,1)$,
\begin{align}\label{eq:def_conditional_phi_NEW}
  \phi_{q,x,q^+}(\rho)&\triangleq  \left( \frac{p_{Y|X}(y|x)}{\pi_{Y|Q}(y|q)}\right)^{-\rho},
\end{align}
where $y$ is the unique solution to the equation $q^+=g(q,y)$. We also use $\phi_{(x,q)^{i}_{i-1}}(\rho)$ as a shorthand for $\phi_{q_{i-1},x_i,q_i}(\rho)$. Throughout the analysis, it is assumed that $p_{Y|X}(y|x)>0$, otherwise, define $\phi_{q,x,q^+}(\rho)=0$ and the derivations can be easily repeated.

Define also
$$
   \psi^s_{x^-,q,x,q^+}(\rho) \triangleq  \E[(S_{i+1}/S_{i})^{-\rho} \mid S_{i}=s, X^-_{i}=x^-,Q_{i-1}=q,X_{i}=x,Q_{i}=q^+].
$$
Indeed, there ought also be a time index in the notation $\psi$, but it will be shown in the proof of Lemma \ref{lemma:bounds_NEW} below that the expected value does not depend on the time index $i$.

\subsection{Analysis}
The following lemma comprises the core of our PMS analysis:
\begin{lemma}\label{lemma:bounds_NEW}
For all $\delta>0$, there exists $s^\ast(\delta)$ such that
\begin{align*}
  \psi^s_{x^-,q,x,q^+}(\rho)&\leq \phi_{q,x,q^+}(\rho)2^\delta,
\end{align*}
for all $s\leq s^\ast(\delta)$, $0\leq \rho<1$, and all $(x^-,q,x,q^+)$.
\end{lemma}
\begin{proof}
In this proof we show that $\psi^s_{x^-,q,x,q^+}(\rho)$ can be made arbitrarily close to $\phi_{q,x,q^+}(\rho)$ if we take $s$ to be small enough. From Lemma \ref{lem:2}, we have
\begin{align}\label{eq:PSI_expected}
  \psi^s_{x^-,q,x,q^+}(\rho) &=  \E[ S_{i+1}^{-\rho}/S_{i}^{-\rho} \mid S_{i}=s, X^-_{i}=x^-,Q_{i-1}=q,X_{i}=x,Q_{i}=q^+]\nn\\
  &= \left(\frac{p_{Y|X}(y|x)}{\pi_{Y|Q}(y|q)}\right)^{-\rho}\E[\tilde{S}_{i}^{-\rho}/S_{i}^{-\rho} \mid S_{i}=s,X^-_{i}=x^-,Q_{i-1}=q,X_{i}=x,Q_{i}=q^+]\nn\\
  &= \phi_{q,x,q^+}(\rho) \E[\tilde{S}_{i}^{-\rho}/S_{i}^{-\rho} \mid S_{i}=s, X^-_{i}=x^-,Q_{i-1}=q,X_{i}=x, Q_{i}=q^+],
\end{align}
where in the second equality above, $y$ is the unique solution of $q^+=g(q,y)$.
Our interest is in showing an upper bound on the expected value in \eqref{eq:PSI_expected}. The simpler case is when $X^-_{i}=1$, since there is no split (i.e., $\tilde{S}_{i}=S_{i}$), so that $\psi^s_{1,q,x,q^+}(\rho)=\phi_{q,x,q^+}(\rho)$.

To deal with the other case, $X^-_{i}=0$, we need the conditional probability density function $f_{\tau_{i},V_{i}\mid S_{i}, Q_{i-1},X^-_{i},X_{i},Q_{i}}(u,v \mid s,q,0,x,q^+)$, where $\tau_{i} \triangleq t_{i}(J_{i}) + U_{i}(q) \!\! \mod \pi_{X^-|Q}(0|q)$ denotes the left end-point of the message interval $J_{i-1}$ after cyclic shifting by $U_{i}(q)$. Observe that
\begin{align}\label{eq:density_UV}
  & f_{\tau_{i},V_{i}\mid S_{i}, Q_{i-1},X^-_{i},X_{i},Q_{i}}(u,v \mid s,q,0,x,q^+) \nn \\
  &\ \ \ \ \ = \ \frac{\Pr[X_i=x,Q_{i}=q^+ \mid S_{i}=s, X^-_{i}=0,Q_{i-1}=q,\tau_{i}=u,V_{i}=v]}{\Pr[X_{i}=x,Q_i=q^+ \mid S_{i}=s,X^-_{i}=0,Q_{i-1}=q]}f_{\tau_{i},V_{i}}(u,v) \nn\\
  &\ \ \ \ \ \stackrel{(a)}= \ \frac{\Pr[X_{i}=x \mid S_{i}=s,X^-_{i}=0,Q_{i-1}=q,\tau_{i}=u,V_{i}=v]}{\Pr[X_i=x \mid S_{i}=s,X^-_{i}=0,Q_{i-1}=q]}f_{U_{i}(q),V_{i}}(u,v)\nn\\
  &\ \ \ \ \ \stackrel{(b)}= \ \frac{\Pr[X_{i}=x \mid S_{i}=s,X^-_{i}=0,Q_{i-1}=q,\tau_{i}=u,V_{i}=v]}{\Pr[X_i=x \mid X^-_{i}=0,Q_{i-1}=q]}f_{U_{i}(q),V_{i}}(u,v)\nn\\
  &\ \ \ \ \ \stackrel{(c)}= \ \frac{\Pr[X_{i}=x \mid S_{i}=s,X^-_{i}=0,Q_{i-1}=q,\tau_{i}=u,V_{i}=v]}{P_{x,0,q}}f_{U,V}\left(\frac{u}{\pi_{X^-|Q}(0|q)},v\right),
\end{align}
where $(a)$ follows from the Markov chain $Q_i-(X_i,Q_{i-1})-(S_{i},X^-_{i},\tau_{i},V_{i})$ and the fact that $\tau_{i}$ is distributed uniformly on $[0,\pi_{X^-|Q}(0|q)]$, $(b)$ follows from the Markov chain $X_{i}-(X^-_{i},Q_{i-1})-S_{i}$ and $(c)$ is due to replacement of the random variable $U_{i}(q)$ with $U_{i}$, and the notation $P_{x,0,q}\triangleq \pi_{X,X^-|Q}(x,0|q)$. The Markov chain $Q_i-(X_i,Q_{i-1})-(S_{i},X^-_{i},\tau_{i},V_{i})$ follows from the fact that $Q_i$ is a function of $(Y_i,Q_{i-1})$ and the memoryless property. The second Markov chain $X_{i}-(X^-_{i},Q_{i-1})-S_{i}$ is shown in two steps: first, if $X^-_{i}=1$, then $X_i=0$. For the other case, $X^-_{i}=0$, note that $W_{i}$, given $X^-_{i}=0$ and $Q_{i-1}=q$, is distributed uniformly on $[0,\pi_{X^-|Q}(0|q)]$, and $X_i$ is a function of $W_i$.

Note that the numerator of \eqref{eq:density_UV} is an indicator function of the event that the point $(u + sv)\bmod \pi_{X^-|Q}(0|q)$ is mapped to $X=x$. Also, we have $f_{U,V}\left(\frac{u}{\pi_{X^-|Q}(0|q)},v\right)=1$ for all $(u,v)\in[0,\pi_{X^-|Q}(0|q)]\times[0,1]$. Thus, the density $f_{\tau_{i},V_{i}\mid S_{i}, Q_{i-1},X^-_{i},X_{i},Q_{i}}$ does not depend on the time index $i$. Since the expected value in \eqref{eq:PSI_expected} is determined from this density function, the time index can be omitted from the notation of $\psi^s_{x^-,q,x,q^+}(\rho)$.

We begin with the calculation of the expected value in \eqref{eq:PSI_expected} for the case $X_i=0$:
\begin{align}\label{eq:expected_value}
   \E & [\tilde{S}_{i}^{-\rho}/S_{i}^{-\rho} \mid S_{i}=s, X^-_{i}=0, Q_{i-1}=q, X_{i}=0, Q_{i}=q^+]\nn\\
   &\stackrel{(a)}= s^\rho \int_{[0,\pi_{X^-|Q}(0|q)]} \int_{[0,1]} f_{\tau_{i},V_{i} \mid S_{i},X^-_{i},Q_{i-1},X_{i},Q_{i}}(u,v \mid s,0,q,0,q^+) \tilde{s}(q,u,s,0)^{-\rho} dv du\nn\\
   &\stackrel{(b)}= \frac{s^\rho}{P_{0,0,q}} \int_{I_1\cup I_2\cup I_3} \int_{[0,1]} \Pr[X_{i}=0 \mid S_{i}=s,X^-_{i}=0,Q_{i-1}=q,\tau_{i}=u,V_{i}=v] \tilde{s}(q,u,s,0)^{-\rho} dv du,
\end{align}
where
\begin{itemize}
  \item[$(a)$] follows by defining $\tilde{s}(q,u,s,0)$ to be the length of the new true message interval after a possible split, which is a function of $(Q_{i-1}=q,\tau_{i} = u, S_{i}=s,X_i = 0)$ only; and
  \item[$(b)$] follows from substituting \eqref{eq:density_UV} and restricting the integration over $u$ to the intervals where $\Pr[X_{i}=0 \mid S_{i}=s,X^-_{i}=0,Q_{i-1}=q,\tau_{i} = u,V_{i}=v]=1$ for some $v \in [0,1]$, i.e., the intervals
\begin{align}\label{eq:intervals}
I_1 &\triangleq \left[\pi_{X,X^-|Q}(1,0|q), \pi_{X^-|Q}(0|q)-s\right]\nn\\
I_2 &\triangleq \left[\pi_{X,X^-|Q}(1,0|q)-s, \pi_{X,X^-|Q}(1,0|q)\right]\nn\\
I_3 &\triangleq \left[\pi_{X^-|Q}(0|q)-s, \pi_{X^-|Q}(0|q) \right],
\end{align}
which are illustrated in Fig. \ref{fig:Navin}.
\end{itemize}

\begin{figure}[!t]
\centering
        \psfrag{V}[][][1]{$v$}
        \psfrag{U}[][][1]{$u$}
        \psfrag{S}[][][0.75]{$I_2$}
        \psfrag{T}[][][0.75]{$I_3$}
        \psfrag{D}[][][0.75]{$I_1$}
        \psfrag{O}[][][.6]{$P_1 - s$}
        \psfrag{P}[][][.6]{$P_1$}
        \psfrag{K}[][][.6]{$P_2-s$}
        \psfrag{L}[][][.6]{$P_2$}
        \psfrag{W}[][][.6]{$1$}
        \centerline{\includegraphics[scale=0.3]{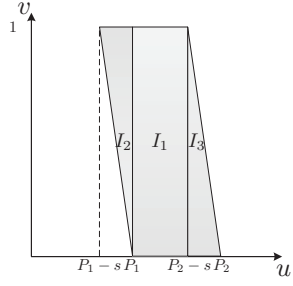}}
\caption{Illustration of the intervals in \eqref{eq:intervals}. We use $P_1 \triangleq  \pi_{X,X^-|Q}(1,0|q)$, $P_2 \triangleq \pi_{X^-|Q}(0|q)$. The shaded area corresponds to values of $u$ and $v$ for which  $\Pr[X_{i}=0 \mid S_{i}=s,X^-_{i}=0,Q_{i-1}=q,\tau_{i}=u,V_{i}=v]=1$.}
\label{fig:Navin}
\end{figure}

For all $u \in I_1$, note that $\tilde{s}(q,u,s,0) = s$, so
\begin{align*}
&\frac{s^\rho}{P_{0,0,q}} \int_{I_1} \int_{[0,1]} \Pr[X_{i}=0 \mid S_{i}=s,X_{i-1}=0,Q_{i-1}=q,\tau_i=u,V_i=v] \tilde{s}(q,u,s,0)^{-\rho} dv du\\
&= \frac1{P_{0,0,q}} |I_1|\\
&= \frac{P_{0,0,q}-s}{P_{0,0,q}}.
\end{align*}

For the interval $I_2$, note from Fig. \ref{fig:Navin} that $\Pr[X_{i}=0|S_{i}=s,X_{i-1}=0,Q_{i-1}=q,\tau_i=u,V_i=v]=1$ only when $v\in\left[\frac{-u+P_{0,0,q}}{s},1\right]$. Therefore,
\begin{align*}
&\frac{s^\rho}{P_{0,0,q}} \int_{I_2} \int_{[0,1]} \Pr[X_{i}=0|S_{i}=s,X_{i-1}=0,Q_{i-1}=q,\tau_i=u,V_i=v] \tilde{s}(q,u,s,0)^{-\rho} dv du\\
&\stackrel{(a)}= \frac{s^\rho}{P_{0,0,q}}\frac1{s} \int_{I_2} \left(u-P_{0,0,q}+s\right) \left(u-P_{0,0,q}+s\right)^{-\rho} du\\
&= \frac{s^\rho}{P_{0,0,q}} \frac1{s} \int_{[0,s]} u^{-\rho+1} du\\
&= \frac1{P_{0,0,q}} \frac{s}{(-\rho+2)},
\end{align*}
where $(a)$ follows from $\tilde{s}(q,u,s,0) = u-P_{0,0,q}+s$. The calculation for the third interval, $I_3$, is similar to that for $I_2$, and results in the same value for the integral.

To conclude, we have shown that
\begin{align*}
\psi^s_{0,q,x,q^+}(\rho)&= \phi_{q,x,q^+}(\rho)\left[ \frac{\pi_{X,X^-|Q}(x,0|q)-s}{\pi_{X,X^-|Q}(x,0|q)} + 2\frac1{\pi_{X,X^-|Q}(x,0|q)} \frac{s}{(-\rho+2)}\right] \\
&\triangleq \phi_{q,x,q^+}(\rho)2^{\delta_{q,x}(s)}
\end{align*}
for $\rho\in[0,1)$. The last step is to define $\delta(s)\triangleq\max_{q,x}\delta_{q,x}(s)$, which goes to zero when $s\to0$. Now, it is clear that $\psi^s_{x^-,q,x,q^+}(\rho)\le\phi_{q,x,q^+}(\rho)2^{\delta(s)}$ for all $(x^-,q,x,q^+)$, $0\le\rho<1$ and $s\le S_{\min}$, as required.
\end{proof}

\begin{proof}[Proof of Theorem \ref{theorem:part_1}]
Throughout this proof, $\zeta_n$ is a sequence of integers that satisfies $\frac{\zeta_n}{n}\to 0$. Our aim is to show that for $|\cM| = 2^{nR}$, there exists $\xi>0$ such that $\Pr[S_{n-\zeta_n} \le \xi]\to 0$ as $n$ increases. For convenience, we replace the variable $n$ by $n + \zeta_n$, so that the analysis is made for a message set of size $|\cM| = 2^{(n+\zeta_n)R}$ and the probability analysis is for  $\Pr[S_{n} \le \xi]$.

From Lemma \ref{lemma:bounds_NEW}, for all $\delta>0$, there exists $s^\ast(\delta)$, such that for all $s\le s^\ast(\delta)$
\begin{align}\label{eq:scheme_fundamental_bound}
  \psi^s_{x^-,q,x,q^+}(\rho)&\leq \phi_{q,x,q^+}(\rho)2^\delta.
\end{align}
We will utilize \eqref{eq:scheme_fundamental_bound} to provide a vanishing upper bound on $\E[\Lambda(S_{n})]$, with $\Lambda(s)=s^{-\rho}$ for a judiciously chosen $\rho > 0$. This, by the Markov inequality, will imply that the probability that $\Lambda(S_n)$ is above a certain threshold is vanishingly small. Since $\Lambda(s)$ is a decreasing function of $s$, we then obtain that the probability that $S_n$ lies below a certain threshold is vanishingly small, as desired.

We introduce some convenient notation for the upcoming analysis. Set $\Delta_i = (X_i,Q_i,S_i)$, and let $(X,Q)^{j}_i$ stand for $(X^{j}_i,Q^{j}_i)$ when $i<j$. Consider the following chain of inequalities:
\begin{align*}
  &\E_{\Delta_n}[\Lambda(S_n)]\\
  &\stackrel{(a)}= \E_{\Delta_{n-1}}[\Lambda(S_{n-1})\E_{\Delta_n|\Delta_{n-1}}[\Lambda(S_n/S_{n-1})|\Delta_{n-1}]] \\
  &\stackrel{}= \E_{\Delta_{n-1}}[\Lambda(S_{n-1})\E_{(X_n,Q_n)|\Delta_{n-1}}[\E_{S_{n}|(X_n,Q_n),\Delta_{n-1}}[\Lambda(S_n/S_{n-1})|(X_n,Q_n),\Delta_{n-1}]]] \\
  &\stackrel{(b)}\leq \E_{\Delta_{n-1}}[\Lambda(S_{n-1})\E_{(X_n,Q_n)|\Delta_{n-1}}[\phi_{(X,Q)_{n-1}^n}(\rho)]] 2^\delta\\
  &\stackrel{}= \E_{(X,Q)_{n-1}^n}[\phi_{(X,Q)_{n-1}^n}(\rho)\E_{S_{n-1}|(X,Q)_{n-1}^n}[\Lambda(S_{n-1})|(X,Q)_{n-1}^n]] 2^\delta\\
  &\stackrel{(a)}= \E_{(X,Q)_{n-1}^n} [\phi_{(X,Q)^{n}_{n-1}}(\rho) \E_{\Delta_{n-2}|(X,Q)_{n-1}^n}[\Lambda(S_{n-2})\E_{S_{n-1}|(X,Q)_{n-1}^n,\Delta_{n-2}}[\Lambda(S_{n-1}/S_{n-2})|(X,Q)_{n-1}^n,\Delta_{n-2}]]]2^\delta\\
  &\stackrel{(c)}= \E_{(X,Q)_{n-1}^n} [\phi_{(X,Q)^{n}_{n-1}}(\rho) \E_{\Delta_{n-2}|(X,Q)_{n-1}^n}[\Lambda(S_{n-2})\E_{S_{n-1}|(X_{n-1},Q_{n-1}),\Delta_{n-2}}[\Lambda(S_{n-1}/S_{n-2})|(X_{n-1},Q_{n-1}),\Delta_{n-2}]]]2^\delta\\
  &\stackrel{(b)}\leq \E_{(X,Q)_{n-1}^n} [\phi_{(X,Q)^{n}_{n-1}}(\rho) \E_{\Delta_{n-2}|(X,Q)_{n-1}^n}[\Lambda(S_{n-2})\phi_{(X,Q)^{n-1}_{n-2}}(\rho)]]2^{2\delta}\\
  &\stackrel{}= \E_{(X,Q)_{n-2}^n} [\phi_{(X,Q)^{n}_{n-1}}(\rho)\phi_{(X,Q)^{n-1}_{n-2}}(\rho) \E_{S_{n-2}|(X,Q)_{n-2}^n}[\Lambda(S_{n-2})|(X,Q)_{n-2}^n]]2^{2\delta}\\
  &\stackrel{(d)}\leq \E_{(X,Q)_{1}^n} \left[\prod_{i=1}^{n}\phi_{(X,Q)^{i}_{i-1}}(\rho) \E_{S_{0}|(X,Q)_{1}^n}[\Lambda(S_{0})|(X,Q)_{1}^n]\right]2^{n\delta}\\
  &\stackrel{(e)}= 2^{(n+\zeta_n)R\rho} 2^{n\delta}\E_{(X,Q)_{1}^n}\left[\prod_{i=1}^{n}\phi_{(X,Q)^{i}_{i-1}}(\rho)\right],
  \end{align*}
  where:
  \begin{itemize}
    \item[(a)] follows from the law of total expectation;
    \item[(b)] follows from \eqref{eq:scheme_fundamental_bound} since $\psi^{S_{n-1}}_{X_{n-1},Q_{n-1},X_{n},Q_{n}}(\rho) = \E_{S_{n}|(X_n,Q_n),\Delta_{n-1}}[\Lambda(S_n/S_{n-1})|(X_n,Q_n),\Delta_{n-1}]$;
    \item[(c)] follows from the Markov chain $S_{i} - (\Delta_{i-1},(X_i,Q_i)) - (X,Q)_{i+1}^n$ for all $i$, and specifically for $i=n-1$;\footnote{This Markov chain follows from the same argument used in Lemma~\ref{lemma:bounds_NEW} for the Markov chain $X_{i+1}-(X_{i},Q_{i})-S_{i}$.}
    \item[(d)] follows from applying the above steps $n-2$ times; and
    \item[(e)] follows from the fact that $|\cM| = 2^{(n+\zeta_n)R}$.
  \end{itemize}
The expectation above can be decomposed into non-typical and typical sequences with respect to the Markov distribution $p(q^+,x|q,x^-):= \sum_y \mathbbm{1}\{q^+=g(q,y)\}p_{Y|X}(y|x)p^*_{X|X^-,Q}(x|x^-,q)$. With some abuse of notation, since $q$ and $q^+$ determine a unique $y$ such that $q^+=g(q,y)$, we refer to $\phi_{(X,Q)^{i}_{i-1}}(\rho)$ as $\phi_{Q_{i-1},X_i,Y_i}(\rho)$. Consider
\begin{align}\label{eq:bound_Expec_summa}
  \E_{\Delta_n}[\Lambda(S_n)]&\stackrel{(a)}\leq 2^{n(R\rho(1 + \frac{\zeta_n}{n})+\delta)}\left[\epsilon_n [\max_{q,x,y} \phi_{q,x,y}(\rho)]^n  + \prod_{q,x,y} \phi_{q,x,y}(\rho)^{n\pi_{Q,X,Y}(q,x,y) + \kappa_n}\right]\nn\\
  &= 2^{n(R\rho + R\rho\frac{\zeta_n}{n}+\delta+K\rho)}\epsilon_n + 2^{n(R\rho + R\rho\frac{\zeta_n}{n}+\delta)} \prod_{q,x,y} 2^{(n\pi_{Q,X,Y}(q,x,y)+\kappa_n)\log \phi_{q,x,y}(\rho)}\nn\\
  &= 2^{n(R\rho + R\rho\frac{\zeta_n}{n}+\delta+K\rho)}\epsilon_n + 2^{n(R\rho + R\rho\frac{\zeta_n}{n}+\delta)} 2^{-\rho \sum_{q,x,y}(n\pi_{Q,X,Y}(q,x,y)+\kappa_n)\log \left( \frac{p_{Y|X}(y|x)}{\pi_{Y|Q}(y|q)}\right)}\nn\\
  &\stackrel{(b)}\le 2^{n(R\rho + R\rho\frac{\zeta_n}{n}+\delta+K\rho)}\epsilon_n + 2^{-n\rho(I(X;Y|Q)-R(1+\frac{\zeta_n}{n})-\frac{\delta}{\rho} + \frac{\kappa'_n}{n})},
  \end{align}
where $(a)$ follows from separating the contributions made to the expected value by non-typical and typical sequences:  we let $\epsilon_n$ denote the probability that a sequence is not in the typical set, $K$ stands for $\max_{q,x,y} \log \left( \frac{\pi_{Y|Q}(y|q)}{p_{Y|X}(y|x)}\right)$ and, finally, $\kappa_n$ denotes the maximum deviation of the empirical distribution from $\pi_{Q,X,Y}(q,x,y)$ for a typical sequence. Item $(b)$ follows from the notation $\kappa'_n\triangleq\kappa_n|\mathcal Q||\mathcal X||\mathcal Y|K$. Now, since $\epsilon_n$ decreases exponentially with $n$, there exists a choice of $(\rho^\ast,\delta^\ast)$ such that $2^{n(R\rho^* + \frac{\zeta_n}{n}+\delta^*+K\rho^*)}\epsilon_n$ is arbitrarily small, while $R$ can be made arbitrarily close to $I(X;Y|Q)$.

Finally, the main result can be derived with $\delta^\ast$ and $\rho^\ast$:
\begin{align}\label{eq:probability}
  \Pr[S_n\leq s^\ast(\delta^\ast)]&\stackrel{(a)}= \Pr[\Lambda(S_n)\geq \Lambda(s^\ast(\delta^\ast))] \nn\\
  &\stackrel{(b)}\leq \frac{\E[\Lambda(S_n)]}{\Lambda(s^\ast(\delta^\ast))} \nn\\
&\stackrel{(c)}\to 0,
\end{align}
 where $(a)$ follows from the fact that $\Lambda(\cdot)$ is a decreasing function, $(b)$ follows from Markov's inequality and $(c)$ follows from \eqref{eq:bound_Expec_summa}.
\end{proof}

\section{DP formulation and solution}\label{sec:DP_formulation}
This section covers the formulation of feedback capacity as DP and its solution. The solution of the DP problem implies almost immediately the derivations of feedback capacity and optimal input distribution, which were stated earlier as separate results in Theorems \ref{theorem:BIBO} and \ref{theorem:optimal_inputs} (which are proved at the end of this section). We begin with presenting the family of DP problems termed infinite-horizon with average reward.
\subsection{Average reward DP}
Each DP is defined by the tuple $(\mathcal{Z},\mathcal{U},\mathcal{W},F,P_Z,P_W,g)$.
We consider a discrete-time dynamical system evolving according to:
\begin{equation}\label{eq_DP}
  z_t=F(z_{t-1},u_t,w_t),\  t=1,2,\dots
\end{equation}
Each state, $z_t$, takes values in a Borel space $\mathcal{Z}$, each action, $u_t$, takes values in a compact subset $\mathcal{U}$ of a Borel space, and each disturbance, $w_t$, takes values in a measurable space $\mathcal{W}$. The initial state, $z_0$, is drawn from the distribution $P_Z$, and the disturbance, $w_t$, is drawn from $P_{W|Z_{t-1},U_t}$.
The history, $h_t=(z_0,w_1,\dots,w_{t-1})$, summarizes all the information available to the controller at time $t$. The controller at time $t$ chooses the action, $u_t$, by a function $\mu_t$ that maps histories to actions, i.e., $u_t = \mu_t(h_t)$. The collection of these functions is called a policy and is denoted as $\pi=\{\mu_1,\mu_2,\dots\}$. Note that given a policy, $\pi$, and the history, $h_t$, one can compute the actions vector, $u^t$, and the states of the system, $z_1,z_2,\dots,z_{t-1}$.

Our objective is to maximize the average reward given a bounded reward function $r: \mathcal{Z}\times \mathcal{U}\rightarrow \mathbb{R}$. The average reward for a given policy $\pi$ is given by:
\begin{equation*}
\rho_{\pi} = \liminf _{N\rightarrow\infty} \frac{1}{N}\mathbb{E}_{\pi}\left[\sum_{t=1}^{N}r(Z_{t-1}, \mu_t(h_t))\right],
\end{equation*}
where the subscript indicates that actions $u_t$ are subject to the policy $\pi$. The optimal average reward is defined as
\begin{equation*}
\rho = \sup_{\pi} \rho_{\pi}.
\end{equation*}

Having defined the DP problem, we are ready to show the formulation of feedback capacity as DP.
\subsection{Formulation of capacity as DP}
The state of the DP, $z_{t-1}$, is defined as the conditioned probability vector $\beta_{t-1}(x_{t-1})\triangleq p(x_{t-1}|y^{t-1})$. The action space, $\mathcal{U}$, is the set of stochastic matrices, $p_{X_t|X_{t-1}}$, such that $p_{X_t|X_{t-1}}(1|1)=0$. For a given policy and an initial state, the encoder at time $t-1$ can calculate the state, $\beta_{t-1}$, since the tuple $y^{t-1}$ is available from the feedback. The disturbance is taken to be the channel output, $w_t=y_{t}$, and the reward gained at time $t-1$ is chosen as $I(Y_{t};X_{t}|y^{t-1})$. These definitions imply that the optimal reward of this DP is equal to the feedback capacity given in Theorem \ref{theorem:cap_as_DP}.

It can also be shown that the DP states satisfy the following recursive relation,
\begin{align}\label{eq:DP_function}
\beta_t(x_t)&=p(x_t|y^{t}) \nonumber\\
            &= \frac{\sum_{x_{t-1}} \beta_{t-1}(x_{t-1})u_t(x_t,x_{t-1})p_{Y|X}(y_{t}|x_t)}{\sum_{x_t,x_{t-1}} \beta_{t-1}(x_{t-1})u_t(x_t,x_{t-1})p_{Y|X}(y_{t}|x_t)},
\end{align}
where $u_t(x_t,x_{t-1})$ corresponds to $p(x_t|x_{t-1},y^{t-1})$, the dependence on $y^{t-1}$ being left out of the notation for $u_t$. In \cite{Sabag_BEC}, this formulation was shown to satisfy the Markov nature required in DP problems and it was also shown that the optimal average reward is exactly the capacity expression in Theorem \ref{theorem:cap_as_DP}. Note that this formulation is valid for any memoryless channel with our input constraint; moreover, minor variations can also yield a similar formulation with different input constraints.

\subsection{The DP for the BIBO channel}\label{subsec:form_for_BC}
Here, each element in the formulation above will be calculated for the BIBO channel; the DP state at time $t-1$, $z_{t-1}$, is the probability vector $[p_{X_{t-1},Y^{t-1}}(0|y^{t-1}), p_{X_{t-1}|Y^{t-1}}(1|y^{t-1})]$. Since the components of this vector sum to $1$, the notation can be abused as $z_{t-1}\triangleq p_{X_{t-1}|Y^{t-1}}(0|y^{t-1})$, i.e., the first component will be the DP state. Each action, $u_{t}$, is a constrained $2\times2$ stochastic matrix, $p_{X_t|X_{t-1}}$, of the form:
\begin{equation*}
u_{t}=\left[\begin{array}{cc}
 p_{X_t|X_{t-1}}(0|0) & p_{X_t|X_{t-1}}(1|0) \\
1 & 0 \end{array}\right].
\end{equation*}
The disturbance $w_t$ is the channel output, $y_t$, and thus, it can take values from $\{0,1\}$.

The notation $\delta_t\triangleq z_{t-1}p_{X_t|X_{t-1}}(1|0)$ is useful and implies the constraint $0\leq\delta_t\leq z_{t-1}$, since $u_t$, by definition, must be a stochastic matrix. Furthermore, given $z_{t-1}$, $u_t$ can be recovered from $\delta_t$ for all $z_{t-1}\neq 0$. For the case $z_{t-1}= 0$, we will see that $p_{X_t|X_{t-1}}(1|0)$ has no effect on the DP, so it can be fixed to zero.
The system equation can then be calculated from \eqref{eq:DP_function}:
\begin{equation}\label{eq:DP_evolution}
z_{t}=\left\{\begin{array}{cc}
 \frac{\bar{\alpha}\bar{\delta}_t}{(1-\alpha)(1-\delta_t) + \beta\delta_t} & \text{if } w_t=0, \\
\frac{\alpha\bar{\delta}_t}{\alpha(1-\delta_t) + (1-\beta)\delta_t} & \text{if } w_t=1. \end{array}\right.
\end{equation}

\begin{table}
\caption{the conditional distribution $p(y_t,x_t,x_{t-1}|z_{t-1},u_t)$}
\label{table:joint_distr}
\centering
\begin{tabular}{|c|c||c|c|}
  \hline
  $x_{t-1}$ & $x_{t}$ & $y_t=0$ & $y_t=1$ \\ \hline \hline
  $0$ & $0$ & $\bar{\alpha}z_{t-1}u_t(1,1)$ & $\alpha z_{t-1}u_t(1,1)$ \\ \hline
  $0$ & $1$ & $\beta z_{t-1}u_t(1,2)$ & $\bar{\beta}z_{t-1}u_t(1,2)$ \\ \hline
  $1$ & $0$ & $\bar{\alpha}(1-z_{t-1})$ & $\alpha(1-z_{t-1})$ \\ \hline
\end{tabular}
\end{table}
The conditional distribution, $p(x_t,x_{t-1},y_t|z_{t-1},u_t)$, is described in Table \ref{table:joint_distr}, so one can calculate the reward
\begin{align*}
r(z_{t-1},u_t) &= I(Y_t;X_t|z_{t-1},u_t) \\
&= H_2(\bar{\alpha}\bar{\delta_t} + \beta\delta_t) - (1-\delta_t) H_2(\alpha) - \delta_t H_2(\beta).
\end{align*}

Before computing the DP operator, it is convenient to define
\begin{equation}\label{eq:def_arg}
\begin{split}
 p_{\alpha,\beta}(\delta)           &= \alpha\bar{\delta} + \bar{\beta}\delta \\
 \text{arg1}_{\alpha,\beta}(\delta) &= \frac{\bar{\alpha}\bar{\delta}}{1 - p_{\alpha,\beta}(\delta)} \\
 \text{arg2}_{\alpha,\beta}(\delta) &= \frac{\alpha\bar{\delta}}{p_{\alpha,\beta}(\delta)}.
 \end{split}
\end{equation}
We will omit the subscripts $\alpha,\beta$ in the notation above when it is clear from the context.
The DP operator is then given by:
\begin{align}\label{eq:DP_operator}
 (Th_{\alpha,\beta})(z) &= \max_{0\leq\delta \leq z} H_2(p(\delta)) - (1-\delta) H_2(\alpha) - \delta H_2(\beta) + (1-p(\delta))h_{\alpha,\beta}(\text{arg1}(\delta)) + p(\delta)h_{\alpha,\beta}(\text{arg2}(\delta)),
\end{align}
for all functions $h_{\alpha,\beta}:[0,1]\rightarrow \mathbb{R}$, parameterized by $(\alpha,\beta)$.

Now when the DP problem for the BIBO channel is well-defined, the Bellman equation which can verify the optimality of rewards, can be used to obtain an analytic solution. However, the Bellman equation cannot be easily solved, and therefore, numerical algorithms are required to estimate the Bellman components. The numerical study of DP problems is not within the scope of this paper, and the reader may find \cite{PermuterCuffVanRoyWeissman08,Ising_channel} to be suitable references for learning this topic in the context of feedback capacities. Therefore, we proceed directly to the statement and the solution of the Bellman equation.

\subsection{The Bellman Equation}
In DP, the Bellman equation suggests a sufficient condition for average reward optimality. This equation establishes a mechanism for verifying that a given average reward is optimal. The next result encapsulates the Bellman equation:
\begin{theorem}[Theorem 6.2, \cite{Arapos93_average_cose_survey}]\label{theorem:bellman}
If $\rho\in\mathbb{R}$ and a bounded function $h:\mathcal{Z}\rightarrow \mathbb{R}$ satisfies for all $z \in \mathcal{Z}$:
\begin{equation}\label{eq_bellman}
    \rho + h(z) = \sup_{u\in\mathcal{U}} r(z,u) + \int p(w|z,u)h(F(z,u,w))dw,
\end{equation}
then $\rho^{\ast}=\rho$. Furthermore, if there is a function $\mu : \mathcal{Z} \rightarrow \mathcal{U}$, such that $\mu(z)$ attains the supremum for each $z$, then $\rho^{\ast}=\rho_\pi$ for $\pi=\{\mu_0, \mu_1,\dots\}$ with $\mu_t(h_t) = \mu(z_{t-1})$ for each $t$.
\end{theorem}
This result is a direct consequence of Theorem $6.2$ in \cite{Arapos93_average_cose_survey}; specifically, the triplet $\left(\rho,h(\cdot),\mu(\cdot)\right)$ is a canonical triplet by Theorem $6.2$, since it satisfies \eqref{eq_bellman}. Now, because a canonical triplet defines for all $N$ the $N$-stage optimal reward and policy under terminal cost $h(\cdot)$, it can be concluded that a canonical triplet also defines the optimal reward and policy in the infinite horizon regime, since in this case, the bounded terminal cost has a negligible effect.

Define the function $R_{\alpha,\beta}:[0,1]\rightarrow \mathbb{R}$:
\begin{align}\label{eq:def_reward}
R_{\alpha,\beta}(z)&= \frac{H_2(\alpha\bar{z} + \bar{\beta}z) + (\alpha\bar{z} + \bar{\beta} z)H_2(\frac{\alpha\bar{\beta}}{\alpha\bar{z} + \bar{\beta} z}) - (\bar{z} + \bar{\beta}z)H_2(\alpha) - (z + \alpha\bar{z})H_2(\beta)}{1 + \alpha\bar{z} + \bar{\beta} z},
\end{align}
and two constants,
\begin{align}\label{eq:def_rho}
  \tilde{\rho}_{\alpha,\beta} &= \max_{0\leq z \leq 1} R_{\alpha,\beta}(z)\nn\\
    z^{\text{opt}}_{\alpha,\beta} &= \argmax_{0\leq z \leq 1} R_{\alpha,\beta}(z).
\end{align}

Also, define the functions:
\begin{align}\label{eq:defintion_X_function}
h_1^{\alpha,\beta}(z)&= H_2(p(z)) - (1-z) H_2(\alpha) - z H_2(\beta) \nn\\
X^{\alpha,\beta}(z)&= H_2(p(z)) - (1-z) H_2(\alpha) - z H_2(\beta) - p(z) \tilde{\rho}_{\alpha,\beta} \nn\\
h_2^{\alpha,\beta}(z)& = \frac{X^{\alpha,\beta}(z) + p(z) X^{\alpha,\beta}(\text{arg2}_{\alpha,\beta}(z))}{1-\alpha\bar{\beta}},
\end{align}
for $z\in[0,1]$.
The concatenation of the above functions can be defined:
\begin{align*}
    \tilde{h}_{\alpha,\beta}(z) =
    \begin{cases}
        h_1^{\alpha,\beta}(z);  &\mbox{if } 0 \leq z \leq z^{\alpha,\beta}_1\\
        h_2^{\alpha,\beta}(z); &\mbox{if } z^{\alpha,\beta}_1 < z \leq z^{\alpha,\beta}_2\\
        \tilde{\rho}_{\alpha,\beta} & \mbox{if } z^{\alpha,\beta}_2 < z \leq 1,
    \end{cases}
\end{align*}
where $z^{\alpha,\beta}_1$ and $z^{\alpha,\beta}_2$ were defined in \eqref{eq:def_Zi}. With these definitions, we are ready to state the fundamental theorem of this section.
\begin{theorem}\label{theorem:sol_bellman}
The function $\tilde{h}_{\alpha,\beta}(z)$ and the constant $\tilde{\rho}_{\alpha,\beta}$ satisfy the Bellman equation, i.e.,
\begin{equation*}
\tilde{h}_{\alpha,\beta} + \tilde{\rho}_{\alpha,\beta} = T\tilde{h}_{\alpha,\beta},
\end{equation*}
for all $[\alpha,\beta]\in[0,1]\times[0,1]$ satisfying $\alpha+\beta\leq 1$. Moreover, the maximum in $T\tilde{h}_{\alpha,\beta}$ is achieved when $\delta^\ast(z)=z$ for $z\in[0,z_2^{\alpha,\beta}]$, and $\delta^\ast(z)=z_2^{\alpha,\beta}$ otherwise.
\end{theorem}

The computations needed to verify that $\tilde{h}_{\alpha,\beta}(z)$ and $\tilde{\rho}_{\alpha,\beta}$ indeed satisfy the Bellman equation are given in Appendix~\ref{app:proof_bellman}. To derive this solution in the first place, we started with the numerical techniques used in similar contexts in the prior works\cite{PermuterCuffVanRoyWeissman08,Ising_channel,Sabag_BEC}. The numerical results indicated that the optimal actions are linear in parts and that the number of visited DP states is $4$. These observations were sufficient for some good guesswork that helped us to find the Bellman equation solution. For further reading on solving the feedback capacity when actions may be non-linear functions but the DP visits a finite number of states, the reader is referred to \cite{Sabag_UB_IT}. Loosely speaking, the paper \cite{Sabag_UB_IT} includes upper and lower bounds on the feedback capacity that match if the DP (under optimal actions) visits a finite number of states.

As consequences of Theorem \ref{theorem:sol_bellman}, we obtain the facts that the feedback capacity and the optimal input distribution of the BIBO channel are as stated in Theorems \ref{theorem:BIBO} and \ref{theorem:optimal_inputs}.

\begin{proof}[Proof of Theorem \ref{theorem:BIBO}]
By Theorem \ref{theorem:sol_bellman}, the DP optimal average reward is $\tilde{\rho}_{\alpha,\beta}$, which is the same as the capacity expression in the statement of Theorem \ref{theorem:BIBO}.
\end{proof}

\begin{proof}[Proof of Theorem \ref{theorem:optimal_inputs}]
We first show that the initial DP state may be assumed to be such that the optimal policy visits only a finite set of DP states. From this, we infer the form of the optimal input distribution given in the theorem statement. It is then straightforward to check that $\{(X_i,Q_i)\}_{i \ge 1}$ forms an irreducible and aperiodic Markov chain. It follows from this that the average reward, i.e., the feedback capacity, can also be expressed as $I(X;Y | Q)$.

Recall the optimal actions from Theorem \ref{theorem:sol_bellman}:
\begin{align*}
    \delta^\ast(z) =
    \begin{cases}
        z  &\mbox{ if } 0 \leq z \leq z^{\alpha,\beta}_2\\
        z_2^{\alpha,\beta} &\mbox{ if } z^{\alpha,\beta}_2 < z \leq 1.
    \end{cases}
\end{align*}
The DP state evolution in \eqref{eq:DP_evolution} can be described using the $\text{arg}j_{\alpha,\beta}$ functions in \eqref{eq:def_arg}. It is easy to check that the set $\mathcal{Z}_*\triangleq \{z^{\alpha,\beta}_i: i=1,2,3,4\}$, with $z^{\alpha,\beta}_i$ as defined in \eqref{eq:def_Zi}, is closed under the composite function $\text{arg}j_{\alpha,\beta} \circ \delta^*$, i.e., $\text{arg}j_{\alpha,\beta}\bigl(\delta^\ast(z^{\alpha,\beta}_i)\bigr)\in\mathcal{Z}_*$ for all $i,j$. The functions $\text{arg}j_{\alpha,\beta}\bigl(\delta^\ast(z)\bigr)$, $j=1,2$, create a sink, meaning that there is always a positive probability for a transition being made from any DP state $z \in \mathcal{Z}$ to a state in $\mathcal{Z}_*$, and zero probability of leaving the set $\mathcal{Z}_*$. Therefore, we can assume that the initial DP state $z_0$ is from $\mathcal{Z}_*$. Note also that for the function $g(\cdot,\cdot)$ that describes the transitions in the $Q$-graph in Fig.~\ref{fig:BC_optimal_input}, we have $\ell=g(i,j-1)$ iff $z_\ell^{\alpha,\beta} = \text{arg}j_{\alpha,\beta}\bigl(\delta^\ast(z_i^{\alpha,\beta})\bigr)$, for all $i,\ell \in \{1,2,3,4\}$ and $j\in\{1,2\}$. Therefore, we may identify the set $\mathcal{Z}_*$ with the set of $Q$-states $\cQ = \{1,2,3,4\}$, so that the evolution of the DP states can be described on the $Q$-graph of Fig.~\ref{fig:BC_optimal_input}. The form given for the optimal input distribution in the theorem statement follows directly from this observation.

We next verify the first-order Markov property of $(X_i,Q_i)_{i\ge1}$. Observe that
\begin{align}\label{eq:first_order}
   p(x_{i},q_{i}|x^{i-1},q^{i-1}) &= \sum_{y_i} p(y_i,x_{i},q_{i}|x^{i-1},q^{i-1}) \nn\\
   &\stackrel{(a)}= \sum_{y_i} p(q_{i}|q_{i-1},y_i)p(y_i|x_{i})p^*_{X|X^-,Q}(x_{i}|x_{i-1},q_{i-1}),
\end{align}
where $(a)$ follows from the structure of $p^*_{X|X^-,Q}$ given in \eqref{eq:transfer_matrices}, the memoryless channel property, and the fact that $Q_i$ is a function of $(Q_{i-1},Y_i)$.
It can be verified that the Markov chain $(X_i,Q_i)_{i\ge1}$ is irreducible and aperiodic, and hence, ergodic. It is then a straightforward, albeit tedious, exercise to verify that its (unique) stationary distribution is given by $\pi_{X^-,Q} = \pi_{X^-|Q}\pi_Q$, as in the statement of the theorem.

Any trajectory of states $z_t$, $t \ge 0$, followed by the DP under the optimal policy accumulates an average reward of $\liminf_N \frac{1}{N} \sum_{t=1}^N r\bigl(z_t,\delta^*(z_t)\bigr)$. Since we assume that $z_0 \in \mathcal{Z}_*$, each term in the sum is of the form $r\bigl(z_i^{\alpha,\beta},\delta^\ast(z_i^{\alpha,\beta})\bigr)$ for some $i \in \cQ$, which is equal to $I(Y;X|Q=i)$, the reward at the DP state $z^{\alpha,\beta}_i$. It then follows from the ergodicity of the Markov chain $(X_i,Q_i)_{i \ge 1}$ that, as $N \to \infty$,  the time average $\frac{1}{N} \sum_{t=1}^N r\bigl(z_t,\delta^*(z_t)\bigr)$ converges almost surely to $\sum_{i=1}^4 I(Y;X|Q=i) \pi_Q(i) = I(X;Y|Q)$.
\end{proof}

\section{Summary and concluding remarks}\label{sec:conclusions}
The capacity of the BIBO channel with input constraints was derived using a corresponding DP problem. A by-product of the DP solution is the optimal input distribution, which can be described compactly using $Q$-graphs. For the S-channel, we were able to derive a capacity-achieving coding scheme with simple and intuitive analysis for the achieved rate. For the general BIBO channel, we provided a PMS construction that includes the new element of history bits that captures the memory embedded in the setting. Furthermore, to ease the analysis of the scheme, a message-interval splitting operation was introduced so as to keep a Markov property of the involved random variables. With these ideas, we showed that the constructed PMS achieves the capacity of the BIBO channel.

The elements that were presented here for the PMS in the input-constrained BIBO channel setting may be exploited to derive a PMS for a broader class of finite-state channels (FSC) with feedback. Specifically, a FSC is \textit{unifilar} if the channel state $s_t$ is a deterministic function of the previous channel state $s_{t-1}$, input $x_t$ and output $y_t$. Though several works have proposed the PMS approach for this class, the assertion that this approach is optimal (in the sense of being feedback-capacity achieving) remains to be proved \cite{achilles_conjecture,achilles_scheme_isit_16}. The idea of history bits that was presented in this paper can be extended to ``history states" for unifilar channels, since the encoder can determine the channel state at each time $t$, assuming knowledge of the initial state $s_0$. Moreover, for all unifilar channels with simple capacity expressions, their optimal input distributions have a $Q$-graph representation \cite{Sabag_UB_IT} so that they satisfy the BCJR-invariant property that is crucial for the PMS construction. Thus, the steps of the construction and the analysis can be repeated in order to show that the corresponding PMS operates successfully in the sense of Theorem \ref{theorem:part_1}.

\appendices
\section{Proof of Lemma \ref{lemma:maximization_domain}}\label{app:maximization_domain}
In this appendix, we first show that the argument that achieves the maximum of
\begin{align*}
  R_{\alpha,\beta}(z) &= \frac{H_2(\alpha\bar{z} + \bar{\beta} z) + (\alpha\bar{z} + \bar{\beta} z)H_2\left(\frac{\alpha\bar{\beta}}{\alpha\bar{z} + \bar{\beta} z}\right) - (\bar{z} + \bar{\beta} z)H_2(\alpha) - (z + \alpha\bar{z})H_2(\beta)}{1 + \alpha\bar{z} + \bar{\beta} z}
\end{align*}
lies within $[z_L,z_U]=\left[\frac{\sqrt{\alpha}}{\sqrt{\alpha}+\sqrt{\bar{\beta}}},\frac{\sqrt{\bar{\alpha}}}{\sqrt{\bar{\alpha}}+\sqrt{\beta}}\right]$.

Let $p(z)=\alpha\bar{z} + \bar{\beta}z$ and denote by $p'$ the derivative of $p(z)$. After some simplifications, the derivative equals:
\begin{align*}
  &\frac{d}{dz} R(z)\\
  &=  \frac1{(1+p(z))^2}\left\{ (1-\alpha\bar{\beta}) (H_2(\alpha)-H_2(\beta)) + (\bar{\beta}-\alpha)\left[2\log (1-p(z)) - \log(p(z)-\alpha\bar{\beta})(1+\alpha\bar{\beta}) + \alpha\bar{\beta}\log \alpha\bar{\beta}\right]\right\}
\end{align*}
The above derivative equals zero when the function
\begin{align*}
  f_{\alpha,\beta}(z)&\triangleq (1-\alpha\bar{\beta}) [H_2(\alpha)-H_2(\beta)] + (\bar{\beta}-\alpha)[2\log (1-p(z)) - \log(p(z)-\alpha\bar{\beta})(1+\alpha\bar{\beta}) + \alpha\bar{\beta}\log \alpha\bar{\beta}]
\end{align*}
equals zero. It is easy to note that the function $f_{\alpha,\beta}(z)$ is a decreasing function of its argument.

We will show two facts:
\begin{align}
  f_{\alpha,\beta}(p(z_L))&\ge0 \label{eq:domain_LB}\\
  f_{\alpha,\beta}(p(z_U))&\le0,\label{eq:domain_UB},
\end{align}
from which we can conclude that $R_{\alpha,\beta}(z)$ attains its maximum at some $z \in [z_L,z_U]$. For the BSC, it needs to be shown that $f_{\alpha,\alpha}(0.5)\le 0\label{eq:domain_BSC}$ which can be verified easily.

We begin with an explicit calculation of $f_{\alpha,\beta}(p(z_L))$:
\begin{align}\label{eq:f_subs_ZL}
    f_{\alpha,\beta}(p(z_L))
    &\stackrel{(a)}=(1-\alpha\bar{\beta}) [H_2(\alpha)-H_2(\beta)] + (\bar{\beta}-\alpha)[2\log (1-\sqrt{\alpha\bar{\beta}}) - \log(\sqrt{\alpha\bar{\beta}}-\alpha\bar{\beta})(1+\alpha\bar{\beta}) + \alpha\bar{\beta}\log \alpha\bar{\beta}]\nn\\
    &= (1-\alpha\bar{\beta})\left[H_2(\alpha)-H_2(\beta) + (\bar{\beta}-\alpha)\log \left(\frac{1-\sqrt{\alpha\bar{\beta}}}{\sqrt{\alpha\bar{\beta}}}\right)\right],
\end{align}
where $(a)$ follows from $p(z_L) = \sqrt{\alpha\bar{\beta}}$. Since $1-\alpha\bar{\beta}\ge0$, we need to show that $\frac{f_{\alpha,\beta}(p(z_L))}{1-\alpha\bar\beta}\ge0$.

We now show that the minimal value of \eqref{eq:f_subs_ZL} is $0$. Consider the first derivative, with respect to $\alpha$, of $\frac{f_{\alpha,\beta}(p(z_L))}{1-\alpha\bar\beta}$:
\begin{align}\label{eq:partial_derivative}
\frac{d}{d\alpha} \left[H_2(\alpha)-H_2(\beta) + (\bar{\beta}-\alpha)\log \left(\frac{1-\sqrt{\alpha\bar{\beta}}}{\sqrt{\alpha\bar{\beta}}}\right)\right]
&=\log\left(\frac{(1-\alpha)\sqrt{\alpha\bar{\beta}}}{\alpha(1-\sqrt{\alpha\bar{\beta}})}\right) - \frac{\bar{\beta}-\alpha}{2\alpha (1-\sqrt{\alpha\bar{\beta}})}\nn\\
& \leq \frac{(1-\alpha)\sqrt{\alpha\bar{\beta}}}{\alpha(1-\sqrt{\alpha\bar{\beta}})} -1  - \frac{\bar{\beta}-\alpha}{2\alpha (1-\sqrt{\alpha\bar{\beta}})} \nn\\
& = \frac{-(\sqrt{\bar{\beta}}-\sqrt{\alpha})^2 }{2\alpha(1-\sqrt{\alpha\bar{\beta}})}\nn\\
& \leq 0,
\end{align}
where the first inequality follows from $\log x < x-1$ for all $x>0$ with $x=\frac{(1-\alpha)\sqrt{\alpha\bar{\beta}}}{\alpha(1-\sqrt{\alpha\bar{\beta}})}$.

Therefore, for each $\beta$, the function is non-increasing in $\alpha$, so the function can only be decreased if we substitute $\alpha=\bar{\beta}$. Since $f_{\bar\beta,\beta}(p(z_L))=0$, inequality \eqref{eq:domain_LB} is proven.

We now use a similar methodology to show \eqref{eq:domain_UB}. The inequality that needs to be shown is
\begin{align}\label{eq:f_subs_ZU}
  f_{\alpha,\beta}(p(z_U))&= (1-\alpha\bar{\beta})[H_2(\alpha)-H_2(\beta)] + ({\bar{\beta}-\alpha})(2\log \sqrt{\bar{\alpha}\beta}  +\alpha\bar{\beta}\log \alpha\bar{\beta} - \log(1-\sqrt{\bar{\alpha}\beta}-\alpha\bar{\beta})(1+\alpha\bar{\beta}))\nn\\
&\le 0.
\end{align}
Because it is difficult to prove straightforwardly, we write inequality \eqref{eq:f_subs_ZU} as a sum of simpler components, i.e., $-f_{\alpha,\beta}(p(z_U))=  F_{\alpha,\beta}^1 +  F_{\alpha,\beta}^2$, and we show that $ F_{\alpha,\beta}^1$ and $ F_{\alpha,\beta}^2$ are always non-negative. The functions are:
\begin{align*}
     F_{\alpha,\beta}^1&= \alpha\bar{\beta}\left[(\bar{\beta}-\alpha)\log\left(\frac{1-\sqrt{\bar{\alpha}\beta}-\alpha\bar{\beta}}{\alpha\bar{\beta}}\right) + H_2(\alpha)-H_2(\beta)\right] \\
     F_{\alpha,\beta}^1&= (\bar{\beta}-\alpha)\log\left(\frac{1-\sqrt{\bar{\alpha}\beta}-\alpha\bar{\beta}}{\bar{\alpha}\beta}\right)- (H_2(\alpha)-H_2(\beta))
\end{align*}

As before, we take the first derivative of $F_{\alpha,\beta}^1$:
\begin{align*}
  \frac{d}{d\alpha}\left[\frac{F_{\alpha,\beta}^1}{\alpha\bar{\beta}}\right]
&= (\bar{\beta}-\alpha)\left[-\frac1{\alpha} + \frac{-\bar{\beta} + \frac{\beta}{2 \sqrt{\bar{\alpha}\beta}}}{1 - \alpha\bar{\beta} -  \sqrt{\bar{\alpha}\beta}}\right] + \log\left(\frac{\bar{\alpha}\bar{\beta}}{1-\sqrt{\bar{\alpha}\beta}-\alpha\bar{\beta}}\right)\\
&\le (\bar{\beta}-\alpha)\left[-\frac1{\alpha} + \frac{-\bar{\beta} + \frac{\beta}{2 \sqrt{\bar{\alpha}\beta}}}{1 - \alpha\bar{\beta} -  \sqrt{\bar{\alpha}\beta}}\right]  + \left(\frac{\bar{\alpha}\bar{\beta}}{1-\sqrt{\bar{\alpha}\beta}-\alpha\bar{\beta}}\right)-1\\
&= \frac{(\sqrt{\bar\alpha}-\sqrt{\beta})\sqrt{\beta}}{2\alpha\sqrt{\bar{\alpha}\beta}(1-\sqrt{\bar{\alpha}\beta}-\alpha\bar{\beta})}[\sqrt{\bar\alpha}(\alpha\sqrt{\beta}-\bar\beta\sqrt{\bar{\alpha}}) + (\alpha\beta-\bar\alpha\bar\beta)]\\
&\le 0,
\end{align*}
where the first inequality follows from $\log x\le x-1$ with $x=\frac{\bar{\alpha}\bar{\beta}}{1-\sqrt{\bar{\alpha}\beta}-\alpha\bar{\beta}}$. The last inequality follows from $\alpha\le\bar{\beta}$, which implies, in turn,
$\alpha\sqrt{\beta}-\bar\beta\sqrt{\bar{\alpha}}\le0$ and $\alpha\beta-\bar\alpha\bar\beta\le0$. We thus conclude that $F_{\alpha,\beta}^1$ is non-increasing in $\alpha$, and therefore, if we take $\alpha$ to be $1-\beta$, we get its minimal value. Note that $F_{\bar{\beta},\beta}^1=0$, so we have $F_{\alpha,\beta}^1\ge0$.

Now we take the derivative of $F_{\alpha,\beta}^2$ with respect to $\beta$:
\begin{align*}
\frac{d}{d\beta} F_{\alpha,\beta}^2
  &= \frac{(\beta-\bar{\alpha}) \sqrt{\bar\alpha\beta} (-\beta + 2 \sqrt{\beta\bar\alpha})}{
     2 \beta^2 (1 - \alpha\bar\beta -\sqrt{\bar\alpha\beta})} + \log\left(\frac{\bar\alpha\bar\beta}{1 - \alpha\bar\beta -\sqrt{\bar\alpha\beta}}\right)\\
  &\le \frac{(\beta-\bar{\alpha}) \sqrt{\bar\alpha} (-\sqrt{\beta} + 2 \sqrt{\bar\alpha})}{
     2 \beta (1 - \alpha\bar\beta -\sqrt{\beta\bar\alpha})} + \left(\frac{\bar\alpha\bar\beta}{1 - \alpha\bar\beta -\sqrt{\bar\alpha\beta}}\right)-1\\
  &= \frac{(\sqrt{\beta}-\sqrt{\bar{\alpha}})}{2 \beta(1 - \alpha\bar\beta -\sqrt{\bar\alpha\beta})} (\bar\alpha\sqrt{\beta} +2\beta\sqrt{\beta} +\sqrt{\bar{\alpha}}(2\bar{\alpha}-\beta))\\
  &\le 0.
\end{align*}
The last inequality follows from $\sqrt{\beta}-\sqrt{\bar{\alpha}}\le0$. Repeating the same steps as was done for $F_{\alpha,\beta}^1$, we find that $F_{\alpha,\beta}^2\ge0$, which, in turn, gives that $-f_{\alpha,\beta}(p(z_U))\ge0$ as required. $\hfill\blacksquare$

\section{Proof of Theorem \ref{theorem:asymptotic}}\label{app:shannon}
Throughout this section, we use $x=\alpha\bar{\alpha}$, and $x', x''$ to stand for the first and second derivatives of $x$, respectively. Recall that $p_\alpha$ in Corollary \ref{coro:BSC} is the solution for $(\alpha\bar{\alpha})\log(\alpha\bar{\alpha}) + 2\log(1-p) = (1+\alpha\bar{\alpha})\log(p-\alpha\bar{\alpha})$, and let $p'_\alpha$ denote its first derivative. The next lemma concerns $p'_\alpha$ and is the foundation for the proof of Theorem \ref{theorem:asymptotic}.
\begin{lemma}\label{lemma:p_derivative}
The first derivative of $p_\alpha$ is:
\begin{align*}
  p'_\alpha&= (1-2\alpha)\frac{(1-p_\alpha)(p_\alpha-\alpha\bar{\alpha})}{(\alpha\bar{\alpha}-1)(1+p_\alpha)}\left[\log \left(\frac{p_\alpha-\alpha\bar{\alpha}}{\bar{\alpha}}\right)-\log \alpha\right] + (1-2\alpha)\frac{1-p_\alpha}{1-\alpha\bar{\alpha}}\\
&\triangleq  K_2(\alpha) - K_1(\alpha)\log \alpha,
\end{align*}
with the defined functions:
\begin{align*}
    K_1(\alpha)&\triangleq  (1-2\alpha)\frac{(1-p_\alpha)(p_\alpha-\alpha\bar{\alpha})}{(\alpha\bar{\alpha}-1)(1+p_\alpha)} \\
    K_2(\alpha)&\triangleq K_1(\alpha) \log \left(\frac{p_\alpha-\alpha\bar{\alpha}}{\bar{\alpha}}\right) + (1-2\alpha)\frac{1-p_\alpha}{1-\alpha\bar{\alpha}}.
\end{align*}
\end{lemma}
Note that $p_0 = 2-\lambda$, so $K_1(\alpha)$ and $K_2(\alpha)$ are defined at $\alpha=0$:
 \begin{align*}
    K_1(0)&=  -\frac{p_0(1-p_0)}{1+p_0} \\
    K_2(0)&= K_1(0) \log p_0 + 1-p_0.
\end{align*}

\begin{proof}[Proof of Lemma \ref{lemma:p_derivative}]
We calculate the first derivative for each side of $2\log(1-p_\alpha) = (1+\alpha\bar{\alpha})\log(p_\alpha-\alpha\bar{\alpha})-(\alpha\bar{\alpha})\log(\alpha\bar{\alpha})$ so we have:
\begin{align}\label{eq:first_side_deriva}
    \frac{-2}{1-p_\alpha}p'_\alpha &= x'\left[\log \left(\frac{p_\alpha-x}{x}\right) - \frac{1+p_\alpha}{p_\alpha-x}\right] + \frac{1+x}{p_\alpha-x}p'_\alpha.
\end{align}
Arranging both sides of \eqref{eq:first_side_deriva} gives the desired equation:
\begin{align*}
p'_\alpha &= (1-2\alpha)\frac{(1-p_\alpha)(p_\alpha-\alpha\bar{\alpha})}{(\alpha\bar{\alpha}-1)(1+p_\alpha)}\log \left(\frac{p_\alpha-\alpha\bar{\alpha}}{\alpha\bar{\alpha}}\right) + (1-2\alpha)\frac{1-p_\alpha}{1-\alpha\bar{\alpha}}.
\end{align*}
\end{proof}

The next lemma is technical and is made to shorten the proof of Theorem \ref{theorem:asymptotic}:
\begin{lemma}\label{lemma:strong_coefficient}
Define $K_3(\alpha)=\frac{x'}{1+p_\alpha}$, then it can be expressed as
\begin{align*}
  K_3(\alpha) &= \frac{1}{1+p_0} + N\alpha\log \alpha + o(\alpha\log \alpha),
\end{align*}
where $N$ is a constant. 
\end{lemma}
The proof of Lemma \ref{lemma:strong_coefficient} appears in Appendix \ref{subsc:strong_coeffi}. We are now ready to prove the main result of this section.
\begin{proof}[Proof of Theorem \ref{theorem:asymptotic}]
Consider the next chain of equalities:
\begin{align}\label{eq:capacity_taylor}
    & C^{\mathrm{BSC}}(\alpha) + H_2(\alpha) +  K_3(\alpha)\alpha\log\alpha \nonumber \\
    & \stackrel{(a)}= \log (1-p_\alpha) -\log (p_\alpha-x) + K_3(\alpha)\alpha\log\alpha\nonumber \\
    &\stackrel{(b)}= \log \left(\frac{1-p_0}{p_0}\right) + \left[ \frac{-p'_\alpha}{1-p_\alpha} - \frac{p'_\alpha-x'}{p_\alpha-x}  + K'_3(\alpha)\alpha\log\alpha + K_3(\alpha)[1+\log\alpha]\right]_{\alpha=0}\alpha + o(\alpha)\nonumber \\
    &\stackrel{(c)}= \log \left(\frac{1-p_0}{p_0}\right) + \left[ p'_\alpha \frac{x-1}{(1-p_\alpha)(p_\alpha-x)} + \frac{x'}{p_\alpha-x}  + K_3(\alpha)[1+\log\alpha]\right]_{\alpha=0}\alpha + o(\alpha)\nonumber \\
    &\stackrel{(d)}= \log \left(\frac{1-p_0}{p_0}\right) +  \left[ [K_2(\alpha)-K_1(\alpha)\log \alpha ] \frac{x-1}{(1-p_\alpha)(p_\alpha-x)} + \frac{x'}{p_\alpha-x}  + K_3(\alpha)[1+\log\alpha]\right]_{\alpha=0} \alpha+ o(\alpha)\nonumber \\
    &\stackrel{(e)}= \log \left(\frac{1-p_0}{p_0}\right) +  M \alpha+  \left[ K_1(\alpha)\log \alpha \frac{1-x}{(1-p_\alpha)(p_\alpha-x)} + K_3(\alpha)\log\alpha\right]_{\alpha=0} \alpha+ o(\alpha)\nonumber \\
    &= \log \left(\frac{1-p_0}{p_0}\right) + M \alpha+ \left[ \frac{-x'}{1+p_\alpha}\log \alpha + K_3(\alpha)\log\alpha \right]_{\alpha=0}\alpha + o(\alpha)\nonumber \\
    &= \log \left(\frac{1-p_0}{p_0}\right) + M \alpha+ o(\alpha)
\end{align}
where:
 \begin{itemize}
   \item[(a)] follows from Corollary \ref{coro:BSC};
   \item[(b)] follows from the Taylor series approximation $f(\alpha) = f(0) + f'(0) \alpha + O(\alpha^2)$;
   \item[(c)] follows from Lemma \ref{lemma:strong_coefficient}, specifically, $\lim_{\alpha\rightarrow 0} K'_3(\alpha)\alpha\log\alpha=0$;
   \item[(d)] follows from Lemma \ref{lemma:p_derivative}, specifically, $p'_\alpha=K_2(\alpha)-K_1(\alpha)\log \alpha$;
   \item[(e)] follows from the notation $M \triangleq K_2(0) \frac{-1}{(1-p_0)(p_0)} + \frac{1}{p_0} + K_3(0)$.
 \end{itemize}
Thus, we have from \eqref{eq:capacity_taylor} that $C^{\mathrm{BSC}}(\alpha) + H_2(\alpha) +  K_3(\alpha)\alpha\log\alpha = \log \left(\frac{1-p_0}{p_0}\right) + M\alpha + o(\alpha)$.

The derivation is completed with the following equalities:
\begin{align*}
  C^{\mathrm{BSC}}(\alpha) &= \log \left(\frac{1-p_0}{p_0}\right) - K_3(\alpha)\alpha\log\alpha - H_2(\alpha) +  M\alpha + o(\alpha) \\
  &\stackrel{(a)}= \log \lambda - [K_3(0) + N\alpha\log\alpha + o(\alpha\log\alpha)]\alpha\log\alpha - H_2(\alpha) +  M\alpha + o(\alpha) \\
  &\stackrel{(b)}= \log \lambda + \frac{2-\lambda}{3-\lambda}\alpha\log\alpha  + \left( \frac{\log(2-\lambda)-(2-\lambda)}{3-\lambda}\right)\alpha + O(\alpha^2\log^2\alpha)
\end{align*}
where:
 \begin{itemize}
   \item[(a)] follows from $p_0 = 2-\lambda$ and Lemma \ref{lemma:strong_coefficient};
   \item[(b)] follows from $H_2(\alpha)= \alpha- \alpha\log\alpha + o(\alpha)$ and arranging the equation.
 \end{itemize}
\end{proof}

\subsection{Proof of Lemma \ref{lemma:strong_coefficient}}\label{subsc:strong_coeffi}
By a Taylor series approximation, we have
\begin{align*}
  &\frac{x'}{1+p_\alpha} - \frac{x'K_1(\alpha)}{(1+p_\alpha)^2}\alpha\log \alpha \\
  &= \frac{1}{1+p_0}  +\left[\frac{x''}{1+p_\alpha} - p'_\alpha\frac{x'}{(1+p_\alpha)^2} - \left(\frac{x'K_1(\alpha)}{(1+p_\alpha)^2}\right)'\alpha\log \alpha - \frac{x'K_1(\alpha)}{(1+p_\alpha)^2}(1+\log \alpha)\right]_{\alpha=0}\alpha + o(\alpha)\\
  &\stackrel{(a)}=  \frac{1}{1+p_0} + C \alpha+\left[ - p'_\alpha\frac{x'}{(1+p_\alpha)^2} - \frac{x'K_1(\alpha)}{(1+p_\alpha)^2}\log \alpha\right]_{\alpha=0}\alpha + o(\alpha)\\
  &\stackrel{(b)}=  \frac{1}{1+p_0}  + C\alpha +\left[K_2(\alpha)\frac{-x'}{(1+p_\alpha)^2} \right]_{\alpha=0}\alpha + o(\alpha)\\
  &\stackrel{(c)}=  \frac{1}{1+p_0}  + \tilde{C} \alpha+ o(\alpha)\\
\end{align*}
where $(a)$ follows from the fact that $\lim_{\alpha\rightarrow 0}\left(\frac{x'K_1(\alpha)}{(1+p_\alpha)^2}\right)'\alpha\log \alpha=0$ and the notation $C=\frac{-2}{1+p_\alpha}- \frac{K_1(0)}{(1+p_0)^2}$, $(b)$ follows from $p'_\alpha=K_2(\alpha)- K_1(\alpha)\alpha\log \alpha$, and finally, $(c)$ follows from the notation $\tilde{C} = C - \frac{K_2(0)}{(1+p_0)^2}$.

So, we have that
\begin{align*}
 \frac{x'}{1+p_\alpha}  &= \frac{1}{1+p_0}  + \tilde{C} \alpha+ o(\alpha) + \frac{x'K_1(\alpha)}{(1+p_\alpha)^2}\alpha\log \alpha.
\end{align*}

Applying the Taylor series approximation once again on $\frac{x'K_1(\alpha)}{(1+p_\alpha)^2}$ gives that:
\begin{align*}
\frac{x'K_1(\alpha)}{(1+p_\alpha)^2}&= \frac{K_1(0)}{(1+p_0)^2} + h(\alpha),
\end{align*}
where $h(\alpha)$ is some function such that $\lim_{\alpha\rightarrow 0}h(\alpha)=0$.

Combining the last two derivations, we have the required equality, i.e.,
\begin{align*}
 \frac{x'}{1+p_\alpha}  &= \frac{1}{1+p_0}  + \tilde{C}\alpha  + o(\alpha) + \left[ \frac{K_1(0)}{(1+p_0)^2} + h(\alpha)\right]\alpha\log \alpha \\
                        &= \frac{1}{1+p_0} + N \alpha\log \alpha + o(\alpha\log \alpha),
\end{align*}
where $N=\frac{K_1(0)}{(1+p_0)^2}$. $\hfill\blacksquare$

\section{Proof of Theorem \ref{theorem:sol_bellman}}\label{app:proof_bellman}
The following lemma is technical and is useful for understanding the structure of $\tilde{h}_{\alpha,\beta}(z)$.
\begin{lemma}\label{lemma:properties_of_h}
For all $[\alpha,\beta]\in[0,1]\times[0,1]$ s.t. $\alpha+\beta\leq 1$,
\begin{enumerate}
  \item The function $\tilde{h}_{\alpha,\beta}(z)$ is continuous on $[0,1]$.
  \item The function $\tilde{h}_{\alpha,\beta}(z)$ is concave on $[0,1]$.
  \item The only maximum of $h_2^{\alpha,\beta}(z)$ is attained at $z=z_2^{\alpha,\beta}$, and its value is $\tilde{\rho}_{\alpha,\beta}$.
  \item The first derivative of $h^{\alpha,\beta}_1(z)$ is non-negative for $z\in[0,z_1^{\alpha,\beta}]$.
\end{enumerate}
\end{lemma}
The proof of Lemma \ref{lemma:properties_of_h} appears in Appendix \ref{appsubsec:proof_lemma_h}.

\begin{proof}[Proof of Theorem \ref{theorem:sol_bellman}]
The function $\tilde{h}_{\alpha,\beta}(z)$ is defined as a concatenation of $h^{\alpha,\beta}_1(z)$, $h^{\alpha,\beta}_2(z)$, and $\tilde{\rho}_{\alpha,\beta}$;
to simplify the calculation of $(T\tilde{h}_{\alpha,\beta})(z)$, the unit interval is partitioned into non-intersecting sub-intervals,
where each sub-interval uniquely determines the function $\tilde{h}_{\alpha,\beta}(\text{argi}(z))$ to be $h^{\alpha,\beta}_1(z)$, $h^{\alpha,\beta}_2(z)$ or $\tilde{\rho}_{\alpha,\beta}$, for $i=1,2$.
Since there are two concatenation points, $z^{\alpha,\beta}_1$ and $z^{\alpha,\beta}_2$, the unit interval is partitioned at the set of points that satisfy,
\begin{align}\label{eq:transition_points}
    \text{arg1}_{\alpha,\beta}(z)&=z^{\alpha,\beta}_i \nonumber\\
    \text{arg2}_{\alpha,\beta}(z)&=z^{\alpha,\beta}_i,
\end{align}
for $i=1,2$.
\begin{figure}[t]
\centering
        \psfrag{A}[b][][.8]{$\text{arg1}_{\alpha,\beta}(z)$}
        \psfrag{B}[t][][.8]{$\text{arg2}_{\alpha,\beta}(z)$}
        \psfrag{C}[t][][.9]{$z^{\alpha,\beta}_1$}
        \psfrag{D}[t][][.9]{$z^{\alpha,\beta}_2$}
        \psfrag{E}[t][][.9]{$z^{\alpha,\beta}_3$}
        \psfrag{F}[t][][.9]{$z^{\alpha,\beta}_4$}
        \centerline{\includegraphics[scale=.4]{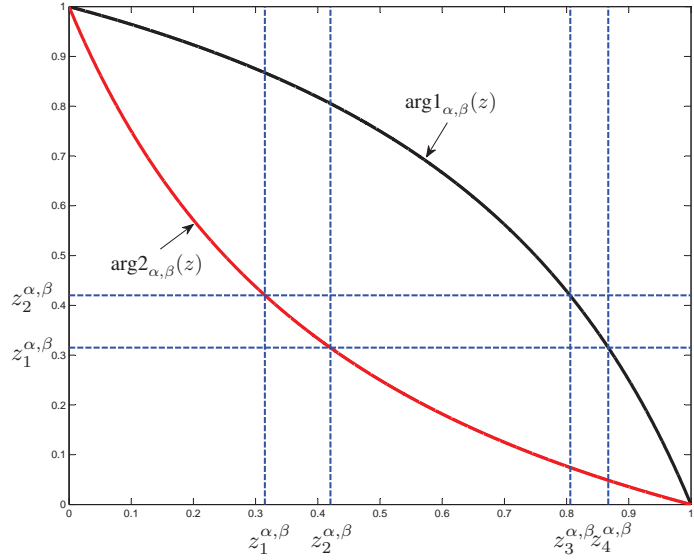}}
\caption{Illustration of the argument functions as a function of $z$, for $\alpha=\beta=0.25$.}
\label{fig:arguments}
\end{figure}
\begin{table}
\caption{The functions $\text{arg1}(z)$ and $\text{arg2}(z)$}
\centering
\begin{tabular}{|c|c|c|c|}
  \hline
   & domain & $\tilde{h}(\text{arg1}(z))$ & $\tilde{h}(\text{arg2}(z))$ \\ \hline\hline
  \RNum{1} & $[0,z_1]$   & $\tilde{\rho}_{\alpha,\beta}$ & $\tilde{\rho}_{\alpha,\beta}$ \\ \hline
  \RNum{2} & $[z_1,z_2]$ & $\tilde{\rho}_{\alpha,\beta}$ & $h_2^{\alpha,\beta}(\text{arg2}(z))$ \\ \hline
  \RNum{3} & $[z_2,z_3]$ & $\tilde{\rho}_{\alpha,\beta}$ & $h_1^{\alpha,\beta}(\text{arg2}(z))$ \\ \hline
  \RNum{4} & $[z_3,z_4]$ & $h_2^{\alpha,\beta}(\text{arg1}(z))$ & $h_1^{\alpha,\beta}(\text{arg2}(z))$ \\ \hline
  \RNum{5} & $[z_4,1]  $ & $h_1^{\alpha,\beta}(\text{arg1}(z))$ & $h_1^{\alpha,\beta}(\text{arg2}(z))$ \\ \hline
\end{tabular}\label{table:arguments}\end{table}

Calculation of the points in \eqref{eq:transition_points} reveals that the unit interval should be partitioned at $z^{\alpha,\beta}_1$, $z^{\alpha,\beta}_2,z^{\alpha,\beta}_3,z^{\alpha,\beta}_4$ from \eqref{eq:def_Zi}. Fig. \ref{fig:arguments} illustrates the argument functions and the partitions when $\alpha=\beta=0.25$. As can be seen from Fig. \ref{fig:arguments}, five segments need to be considered when calculating $(T\tilde{h}_{\alpha,\beta})(z)$. The relevant segments are summarized in Table \ref{table:arguments} together with $\tilde{h}_{\alpha,\beta}(\text{argi}(z))$ for $i=1,2$ for each sub-interval.

Now, the operator $T\tilde{h}_\alpha(z)$ can be calculated, such that in each calculation, we restrict actions to one sub-interval from Table \ref{table:arguments}. For the interval \RNum{1}, i.e., $z\in[0,z_1^{\alpha,\beta})$,
\begin{align}\label{eq:first_interval}
&(T\tilde{h}_{\alpha,\beta})(z)\nonumber\\
&= \sup_{0 \leq \delta \leq z} H_2(p(\delta)) - (1-\delta) H_2(\alpha) - \delta H_2(\beta) + (1-p(\delta))h_{\alpha,\beta}(\text{arg1}(\delta)) + p(\delta)h_{\alpha,\beta}(\text{arg2}(\delta))\nonumber\\
&\stackrel{(a)}= \sup_{0\leq\delta \leq z} H_2(p(\delta)) - (1-\delta) H_2(\alpha) - \delta H_2(\beta)  + (1-p(\delta))\tilde{\rho}_{\alpha,\beta} + p(\delta)\tilde{\rho}_{\alpha,\beta} \nonumber\\
&\stackrel{(b)}= \sup_{0\leq\delta \leq z} h_1(\delta)  + \tilde{\rho}_{\alpha,\beta} \nonumber\\
&\stackrel{(c)}= h_1(z)  + \tilde{\rho}_{\alpha,\beta}
\end{align}
where $(a)$ follows from the restriction of $z\in[0,z_1]$ and substituting the functions from Table \ref{table:arguments}, $(b)$ follows from the definition of $h_1(\delta)$ in \eqref{eq:defintion_X_function} and $(c)$ follows from Item 4) of Lemma \ref{lemma:properties_of_h}, specifically, $h^{\alpha,\beta}_1(z)$ is non-decreasing on $[0,z^{\alpha,\beta}_1]$. Note that the maximizer is $\delta(z)=z$.

The operator with actions restricted to interval \RNum{2}, i.e., $\delta\in[z^{\alpha,\beta}_1,z]$ for $z\in[z_1^{\alpha,\beta},z_2^{\alpha,\beta}]$ is:
\begin{align}\label{eq:second_interval}
    &\sup_{z_1\leq\delta \leq z} H_2(p(\delta)) - (1-\delta) H_2(\alpha) - \delta H_2(\beta) + (1-p(\delta))\tilde{h}_{\alpha,\beta}(\text{arg1}(\delta)) + p(\delta)\tilde{h}_{\alpha,\beta}(\text{arg2}(\delta))\nonumber\\
    &\stackrel{(a)}= \sup_{z_1\leq\delta \leq z} X(\delta) + p(\delta)\tilde{\rho}_{\alpha,\beta}  + (1-p(\delta))\tilde{\rho}_{\alpha,\beta} + \left[\frac{p(\delta) X(\text{arg2}(\delta)) + \alpha\bar{\beta}X(\delta)}{1-\alpha\bar{\beta}}\right]\nonumber\\
    &\stackrel{(b)}= \sup_{z_1\leq\delta \leq z} h_2(\delta) + \tilde{\rho}_{\alpha,\beta}\nonumber\\
    &\stackrel{(c)}= h_2(z) + \tilde{\rho}_{\alpha,\beta},
\end{align}
where $(a)$ follows from the definition of $X(\delta)$ in \eqref{eq:defintion_X_function} and Table \ref{table:arguments}, $(b)$ follows the expression for
$h_2(\delta)$ in \eqref{eq:defintion_X_function} and $(c)$ follows from Item 3) in Lemma \ref{lemma:properties_of_h}, where it was shown that $h_2(z)$ is increasing on $[0,z_2^{\alpha,\beta}]$.

To conclude the calculation of $(T\tilde{h}_{\alpha,\beta})(z)$ for $z\in[z^{\alpha,\beta}_1,z^{\alpha,\beta}_2]$, consider
\begin{align}\label{eq:operator_second}
    (T\tilde{h}_{\alpha,\beta})(z) &\stackrel{(a)}= \max(\sup_{z\in[0,z_1]} h_1(z) + \tilde{\rho}_{\alpha,\beta}, \sup_{z\in[z_1,z]} h_2(z) + \tilde{\rho}_{\alpha,\beta} ) \nonumber\\
    &\stackrel{(b)}=\max(h_1(z_1), h_2(z) )+ \tilde{\rho}_{\alpha,\beta}\nonumber \\
    &\stackrel{(c)}=h_2(z) + \tilde{\rho}_{\alpha,\beta},
\end{align}
where $(a)$ follows from \eqref{eq:first_interval} and \eqref{eq:second_interval}, and both $(b)$ and $(c)$ follow from Items 3) and 4) in Lemma \ref{lemma:properties_of_h}. Note also here that the maximizer of $(T\tilde{h}_{\alpha,\beta})(z)$ for $z$ on sub-interval II is $\delta(z)=z$.

For actions that are restricted to interval \RNum{3}, i.e., $\delta\in[z_2^{\alpha,\beta},z]$ with $z\in[z^{\alpha,\beta}_2,z_3^{\alpha,\beta}]$, consider
\begin{align}\label{eq:third_interval}
& \sup_{z_2 \leq\delta \leq z} H_2(p(\delta)) - (1-\delta) H_2(\alpha) - \delta H_2(\beta) + (1-p(\delta))\tilde{h}_{\alpha,\beta}(\text{arg1}(\delta)) + p(\delta)\tilde{h}_{\alpha,\beta}(\text{arg2}(\delta))\nonumber\\
&\stackrel{(a)}= \sup_{z_2 \leq\delta \leq z} X(\delta) + \tilde{\rho}_{\alpha,\beta} + p(\delta) \left[X(\text{arg2}_{\alpha,\beta}(\delta)) + \left(\frac{\alpha\bar{\beta}}{p(\delta)}\right)\tilde{\rho}_{\alpha,\beta}\right]\nonumber\\
&\stackrel{(b)}= \tilde{\rho}_{\alpha,\beta} + \tilde{\rho}_{\alpha,\beta},
\end{align}
where $(a)$ follows from the definition of $X(\delta)$ in \eqref{eq:defintion_X_function} and Table \ref{table:arguments} and $(b)$ follows from Item 3) in Lemma \ref{lemma:properties_of_h}, specifically, $h_2(z)$ achieves its maximum value at $z=z_2$. Note from \eqref{eq:operator_second} and \eqref{eq:third_interval} that the operator on III satisfies  $(T\tilde{h}_{\alpha,\beta})(z)= 2\tilde{\rho}_{\alpha,\beta}$ and that the maximizer is $\delta(z)=z_2$.

For the action restricted on interval \RNum{4}, i.e., $\delta\in[z_3,z]$ with $z\in[z_3,z_4]$, consider
\begin{align}\label{eq:fourth_interval}
& \sup_{z_3 \le\delta \leq z} H_2(p(\delta)) - (1-\delta) H_2(\alpha) - \delta H_2(\beta) + (1-p(\delta))\tilde{h}_{\alpha,\beta}(\text{arg1}(\delta)) + p(\delta)\tilde{h}_{\alpha,\beta}(\text{arg2}(\delta))\nonumber\\
&\stackrel{(a)}= \sup_{z_3 \leq\delta \leq z} X(\delta) + p(\delta)\tilde{\rho}_{\alpha,\beta} + (1 - p(\delta)) h_2(\text{arg1}_{\alpha,\beta}(\delta))+ p(\delta) \left[X(\text{arg2}(\delta)) + \frac{\alpha\bar{\beta}}{p(\delta)}\tilde{\rho}_{\alpha,\beta}\right]\nonumber\\
&\stackrel{(b)}\leq \tilde{\rho}_{\alpha,\beta} + \tilde{\rho}_{\alpha,\beta},
\end{align}
where $(a)$ follows from the definition of $X(\delta)$ in \eqref{eq:defintion_X_function} and Table \ref{table:arguments} and $(b)$ follows from $h_2(z) \leq \tilde{\rho}_{\alpha,\beta}$ shown in Item 3), Lemma \ref{lemma:properties_of_h}.

The calculation of the last interval, $[z_4,1]$, is omitted here, but it follows the same repeated arguments, so we have $(T\tilde{h}_{\alpha,\beta})(z)\leq 2\tilde{\rho}_{\alpha,\beta}$.
Now, Item 3) in Lemma \ref{lemma:properties_of_h} together with \eqref{eq:fourth_interval} gives us that $(T\tilde{h}_{\alpha,\beta})(z)= 2 \tilde{\rho}_{\alpha,\beta}$ also for $z\in[z_2,z_4]$.
To conclude, we have shown that $(T\tilde{h}_{\alpha,\beta})(z) = \tilde{h}_{\alpha,\beta}(z) + \tilde{\rho}_{\alpha,\beta}$.

\subsection{Proof of Lemma \ref{lemma:properties_of_h}}\label{appsubsec:proof_lemma_h}
Throughout this section, we use $z_i$ as shorthand for $z_i^{\alpha,\beta}$ and $p^{\text{opt}}$ stands for $p(z^{\alpha,\beta}_2)$.

\textbf{Continuity:} Each of the functions defining $\tilde{h}_{\alpha,\beta}(z)$ is continuous, and therefore, one should verify that the concatenation points satisfy
\begin{align}\label{eq:lemma_continuity}
h_1^{\alpha,\beta}(z_1) &= h_2^{\alpha,\beta}(z_1) \nonumber\\
h_2^{\alpha,\beta}(z_2) &=\tilde{\rho}_{\alpha,\beta}.
\end{align}

The second equality in \eqref{eq:lemma_continuity} is verified as follows:
\begin{align*}
 &(1-\alpha\bar{\beta})h^{\alpha,\beta}_2(z_2) \\
 &= X(z_2) + p^{\text{opt}}X(\text{arg}_2(z_2)) \\
 &=H_2(p^{\text{opt}}) + p^{\text{opt}}H_2\left(\frac{\alpha\bar{\beta}}{p^{\text{opt}}}\right) - (\bar{z}_2 + \bar{\beta}z_2)H_2(\alpha) - (z_2 + \alpha\bar{z}_2)H_2(\beta)  - (p^{\text{opt}} + \alpha\bar{\beta})\tilde{\rho}_{\alpha,\beta} \\
 &= (1-\alpha\bar{\beta})\tilde{\rho}_{\alpha,\beta},
\end{align*}
and since $(1-\alpha\bar{\beta}) \neq 0$, it follows that $h_2^{\alpha,\beta}(z_2) =\tilde{\rho}_{\alpha,\beta}$. This derivation also gives us that
\begin{align}\label{eq:h2_at_z_relation}
  \tilde{\rho}_{\alpha,\beta} &= h_2^{\alpha,\beta}(z_2) \nonumber \\
   & \stackrel{(a)}= \frac{X(z_2) + p^{\text{opt}} X(z_1)}{1-\alpha\bar{\beta}},
\end{align}
where $(a)$ follows from the fact that $z_1=\text{arg2}(z_2)$.

The value of $h_2^{\alpha,\beta}(z_1)$ is
\begin{align}\label{eq:h2_at_z1}
  h_2^{\alpha,\beta}(z_1) &= \frac{X(z_1) + \frac{\alpha\bar{\beta}}{p^{\text{opt}}}X(z_2)}{1-\alpha\bar{\beta}} \nonumber\\
         &\stackrel{(a)}= \frac{1}{p^{\text{opt}}}[\tilde{\rho}_{\alpha,\beta} - X(z_2)] \nonumber\\
         &= H_2\left(\frac{\alpha\bar{\beta}}{p^{\text{opt}}}\right) - \frac{\bar{\beta}z_2}{p^{\text{opt}}}H_2(\alpha) - \frac{\alpha\bar{z}_2}{p^{\text{opt}}}H_2(\beta) \nonumber\\
         &= H_2\left(\frac{\alpha\bar{\beta}}{p^{\text{opt}}}\right) - \bar{z_1}H_2(\alpha) - z_1H_2(\beta)
\end{align}
where $(a)$ follows from \eqref{eq:h2_at_z_relation}.

From \eqref{eq:defintion_X_function}, we have that
\begin{align}\label{eq:h1_at_s1}
      h_1^{\alpha,\beta}(z_1)&= H_2\left(\frac{\alpha\bar{\beta}}{p^{\text{opt}}}\right) - \bar{z_1}H_2(\alpha) - z_1H_2(\beta) ],
\end{align}
which together with \eqref{eq:h2_at_z1} concludes the continuity of $\tilde{h}_{\alpha,\beta}(z)$ at $z=z_1$.

\textbf{Concavity:}
First, we show that each element in $\tilde{h}_{\alpha,\beta}(z)$ is concave and then we argue that the concatenation of these functions is also concave. The function $h_1^{\alpha,\beta}(z)$ is concave since it is a composition of the binary entropy function, which is concave, with an affine function. The function $h^{\alpha,\beta}_2(z)$ can be written explicitly from \eqref{eq:defintion_X_function}, and then all of its elements are linear except for the entropy function which is concave and the expression $p(z)H_2\left(\frac{\alpha\bar{z}}{p(z)}\right)$. The latter expression is also concave since it is known that the perspective of the concave function $H_2(z)$, that is, $tf\left(\frac{x}{t}\right)$ for $t>0$ is also concave. Therefore, each element comprises $\tilde{h}_{\alpha,\beta}(z)$ is concave.

It was shown in \cite[Lemma 5]{Ising_channel} that a continuous concatenation of concave functions is concave if the one-sided derivatives at the concatenation points are equal. Therefore, $\tilde{h}_{\alpha,\beta}(z)$ is concave if the following conditions are satisfied:
\begin{align}
     h_1'(z_1) &= h_2'(z_1) \label{concavity_cond1}\\
     h_2'(z_2) &= 0.\label{concavity_cond2}
\end{align}

An auxiliary relation is derived by using the derivative of $R_{\alpha,\beta}(z)$
 \begin{align}\label{eq:derivative_R}
   \frac{d}{dz} R_{\alpha,\beta}(z) &= \frac{(p'H'_2(p) + \left[p H_2\left(\frac{\alpha\bar{\beta}}{p}\right)\right]' + \beta H_2(\alpha)- \bar{\alpha} H_2(\beta))(1+p) - p'(1+p)R_{\alpha,\beta}(z)}{(1+p)^2} \nonumber\\
                     &= \frac{p'H'_2(p) + \left[p H_2\left(\frac{\alpha\bar{\beta}}{p}\right)\right]' + \beta H_2(\alpha)- \bar{\alpha} H_2(\beta) - p'R_{\alpha,\beta}(z)}{1+p},
 \end{align}
where $p'$ is first derivative of $p(z)$. Since $z_2^{\alpha,\beta}$ is the maximum of $R_{\alpha,\beta}(z)$, the numerator of \eqref{eq:derivative_R} is equal to zero at this point (one can verify that $R_{\alpha,\beta}(z)=0$ when $z=0$ or $z=1$), and one can obtain the relation
\begin{align}\label{eq:relation_reward}
p'H'_2(p^{\text{opt}}) + \left[p H_2\left(\frac{\alpha\bar{\beta}}{p}\right)\right]'\vline_{p=p^{\text{opt}}} = - \beta H_2(\alpha)+ \bar{\alpha} H_2(\beta)  + p'\tilde{\rho}_{\alpha,\beta}
\end{align}

The following calculations are also necessary:
\begin{align*}
  X'(z) &= p'H_2'(p) + H_2(\alpha) - H_2(\beta) -p'\tilde{\rho}_{\alpha,\beta}\\
  X(\text{arg2}(z)) &= H_2\left(\frac{\alpha\bar{\beta}}{p}\right) - \overline{\text{arg2}(z)}H_2(\alpha) - \text{arg2}(z)H_2(\beta) - \frac{\alpha\bar{\beta}}{p} \tilde{\rho}_{\alpha,\beta}\\
  X'(\text{arg2}(z)) &= p'H_2'\left(\frac{\alpha\bar{\beta}}{p}\right) + H_2(\alpha) - H_2(\beta) -p'\tilde{\rho}_{\alpha,\beta}
\end{align*}
where derivatives are taken with respect to $z$.

The first derivative of $h_2(z)$ is:
\begin{align}\label{eq:derivative_h2}
   &\frac{d}{dz} ((1-\alpha\bar{\beta})h^{\alpha,\beta}_2(z)) \nn\\
   &= X'(z) + p' X(\text{arg2}(z)) + p \text{arg2}'(z)X'(\text{arg2}(z)) \nn\\
   &= p'H_2'(p) + p'H_2\left(\frac{\alpha\bar{\beta}}{p}\right) -\frac{\alpha\bar{\beta}}{p}p' H_2'\left(\frac{\alpha\bar{\beta}}{p}\right)+ H_2(\alpha) - H_2(\beta) -p'\tilde{\rho}_{\alpha,\beta} \nn\\
   & + p'\left[ - \overline{\text{arg2}(z)}H_2(\alpha) - \text{arg2}(z)H_2(\beta) \right] -\frac{\alpha\bar{\beta}}{p}\left[H_2(\alpha) - H_2(\beta)\right].
\end{align}
Substituting $z=z_2$ into \eqref{eq:derivative_h2} and using \eqref{eq:relation_reward} we obtain the desired equality
\begin{align*}
   (1-\alpha\bar{\beta})\frac{d}{dz} h^{\alpha,\beta}_2(z)\vline_{z=z_2} &= 0.
\end{align*}

For the other condition, \eqref{concavity_cond2}, one can show that $p(z^{\alpha,\beta}_1) = \frac{\alpha\bar{\beta}}{p(z_2)}$, which results in
\begin{align*}
(1-\alpha\bar{\beta})\frac{d}{dz} h^{\alpha,\beta}_2(z)\vline_{z=z_1} &= (1-\alpha\bar{\beta})[p'H_2'\left(\frac{\alpha\bar{\beta}}{p^{\text{opt}}}\right) + H_2(\alpha) - H_2(\beta)]
\end{align*}

The derivative of $h_1^{\alpha,\beta}(z)$ at $z=z_1$ is:
\begin{align*}
  \frac{d}{dz} h^{\alpha,\beta}_1(z_1) &= p' H_{2}'\left(\frac{\alpha\bar{\beta}}{p^{\text{opt}}}\right)  + H_2(\alpha) - H_2(\beta),
\end{align*}
and this concludes the concavity of $\tilde{h}_{\alpha,\beta}(z)$.

The last two items in Lemma \ref{lemma:properties_of_h} follow from the concavity of $\tilde{h}_{\alpha,\beta}(z)$ and the fact that $z_1\le z_2$: since the maximum is at $z_2$, then the derivative of $h_2^{\alpha,\beta}(z)$ at $z_1$ is positive and equals the derivative of $h_1^{\alpha,\beta}(z)$ at $z_1$.
\end{proof}
\section*{Acknowledgment}
The authors would like to thank the Associate Editor and the anonymous reviewers for their valuable and constructive comments, which helped to improve the paper and its presentation considerably.
\bibliography{ref}
\bibliographystyle{IEEEtran}
\end{document}